\documentclass[12pt,oneside]{amsart}
\pdfoutput=1
\usepackage{euscript,amsmath,amsthm,amsfonts,amssymb}
\usepackage[arc,curve,knot,matrix,pdf]{xy}
\usepackage{graphicx}
\usepackage{amsaddr}

\usepackage{color}

\numberwithin{equation}{section}
\numberwithin{figure}{section}

\newtheorem{theorem}{Theorem}[section]
\newtheorem{lemma}[theorem]{Lemma}
\newtheorem{fact}[theorem]{Fact}
\newtheorem{corollary}[theorem]{Corollary}
\newtheorem{claim}[theorem]{Claim}
\newtheorem{observation}[theorem]{Observation}
\newtheorem{remark}[theorem]{Remark}
\newtheorem{note}[theorem]{Note}
\newtheorem*{definition*}{Definition}

\DeclareMathOperator{\DS}{DS}
\DeclareMathOperator{\GDS}{GDS}
\DeclareMathOperator{\TC}{TC}
\DeclareMathOperator{\GTC}{GTC}
\DeclareMathOperator{\Hom}{Hom}
\DeclareMathOperator{\coker}{coker}
\DeclareMathOperator{\image}{image}

\DeclareMathOperator{\DW}{DW}
\DeclareMathOperator{\TDW}{TDW}

\newcommand{\C}{\mathbb{C}}

\DeclareMathAlphabet{\mathpzc}{OT1}{pzc}{m}{it}

\newcommand{\be}{\begin{equation}}
\newcommand{\ee}{\end{equation}}

\begin{document}

\title{Double semions in arbitrary dimension}

\author{Michael H. Freedman$^{1,2}$ \and Matthew B.~Hastings$^1$}
\address{$^1$Station Q, Microsoft Research, Santa Barbara, CA 93106-6105, USA}
\address{$^2$Department of Mathematics, University of California, Santa Barbara, CA 93106 USA}

\begin{abstract}
We present a generalization of the double semion topological quantum field theory to higher dimensions, as a theory of $d-1$ dimensional surfaces in a $d$ dimensional ambient
space.  We construct a local Hamiltonian which is a sum of commuting projectors and analyze the excitations and the ground state degeneracy.  Defining a consistent set of local rules requires the sign structure of the ground state wavefunction to depend not just
on the number of disconnected surfaces, but also upon their higher Betti numbers through the {\it semicharacteristic}.  For odd $d$ the theory is related to the
toric code by a local unitary transformation, but for even $d$ the dimension of the space of zero energy ground states is in general different from the toric code and for even $d>2$ it is also in general different from that of the twisted $Z_2$ Dijkgraaf-Witten model.
\end{abstract}
\maketitle

\section{}
The notion of a topological quantum field theory (TQFT) had its initial success as a $2+1$ dimensional theory employed by Witten to give the Jones polynomial a coordinate free interpretation.  There are a \emph{few} regimes in which TQFTs generalize straightforwardly to higher dimension, $d+1$, $d>2$.  These include Dijkgraaf-Witten theories \cite{DW} or, more generally, theories based on the category of maps to a space $Y$ with only finitely many nontrivial homotopy groups which are, themselves, finite \cite{Q}.  Another family of examples are generalized toric codes ($\GTC$s): Given $0<p < d$ there is a $d+1$ dimensional TQFT whose Hilbert space on a closed (not necessarily oriented) $d$-manifold $X$ is spanned by the elements of $H_p\left(X; Z_2\right)$.
See Ref.~\cite{K} for the case $p=1,d=2$, see Ref.~\cite{systolic} for the case $p=1,d=3$, and see Ref.~\cite{dennis} for the case
$p=2,d=4$.

In dimension $2+1$ the toric code is often studied hand-in-hand with a sister theory ``double semions'' or $\DS$.  In this theory the ground state vectors are multi-loops, with a loop fugacity of $\delta = -1$.  The purpose of this paper is to explain (only) in the case $p = d-1$ how the $\DS$ theory can be generalized to $d>2$.  We call this the generalized (double) semion theory ($\GDS$); we parenthesize the word ``double" because these theories are not
doubles in any sense that we know.  We address:
\begin{enumerate}
    \item gapped local Hamiltonian.  This Hamiltonian will be a sum of commuting local terms, each of which is a projector.
    \item The form of ground state wave functions---answering the question: What is the generalization of loop fugacity $\delta = -1$?
    \item The nature of loop and dual loop operations, which we now call ``balloon'' and ``dual loop.', and also the nature of excitations.
    \item To what class of $d$-manifolds $\{X^d\}$ should we restrict, to arrive at a theory still analogous to the $\GTC$s, insofar as ground space degeneracies and balloon and loop operators are concerned.
\end{enumerate}

The multi-loops of the theory will be defined on a fixed, finite cellulation of some ambient manifold.  Thus, the Hilbert space of the theory will be
finite dimensional.  Hamiltonians where the terms are commuting projectors are often studied as lattice model realizations of a TQFT.  Such lattice model realizations are known in $d=2$ via the
Levin-Wen construction\cite{LW} and in $d=3$ via the Walker-Wang construction\cite{WW}.
Lattice models where the terms
in the Hamiltonian commute are mathematically interesting as these properties simplify the mathematical treatment of these theories; further, stabilitiy
of these Hamiltonians under small perturbations has been proven under fairly general assumptions\cite{stab1,stab2}, so that results about topological
order in these idealized Hamiltonians can often be shown to hold for more general Hamiltonians.

\section{Brief review of toric code ($\TC$) and double semions ($\DS$)}
The generalized toric code model takes as input a cellulation ${\mathcal C}$ of a $d$ dimensional manifold $X$.  There is a qubit degree of freedom for each $p$ cell of ${\mathcal C}$, for some given $p$ which is also input to the theory,
with the Hilbert space ${\mathcal H}$ being the tensor product of these Hilbert spaces (generalizations exists to qudit degrees of freedom; we do not
consider this here).
The Hamiltonian acting on ${\mathcal H}$ is a sum of two terms,
\be
H=H_+ + H_\square,
\ee
where
\be
H_+ = \sum_{\text{$p-1$ cells }e} H_e\text{, }H_\square = \sum_\text{$p+1$ cells $c$} H_c.
\ee
We define
\be
H_e=\frac{1-\prod_{\text{$p$ cells $f$ s.t. $e \in \partial f$}} Z_f}{2},
\ee
where $Z_f$ is the Pauli-$Z$ operator on the cell $f$ and $\partial f$ denotes the boundary of $f$.
We define
\be
H_c=\frac{1- \prod_{\text{$(d-1)$ cells $f \in \partial c$}} X_{f}}{2},
\ee
where the operator $X_{f}$ is the Pauli-$X$ operator acting on the cell $f$, changing the state of the cell.
When we describe product of several Pauli-$X$ operators on some set of cells later, we will sometimes refer to this as $\operatorname{NOT}$ on that set.

These terms are all projectors, and they pairwise commute.
The ground state degeneracy is given by $2^{H_p(X,Z_2)}$.
A basis for the ground state wavefunctions of this theory is to take equal amplitude superpositions of closed $p$-chains in a given homology class.

The theory is a topological quantum field theory for $0<p<d$.
We now restrict to $d=2,p=1$.
On a sphere, the ground state wavefunction of the toric code is a sum over equal amplitude superpositions of closed $1-$chains.
The double semion model on a sphere has a similar ground state wavefunction up to a sign factor: the ground state wavefunction is a sum over closed
$1$-chains, with an amplitude equal to $-1$ raised to a power equal to the number of closed loops.
To define this amplitude unambiguously, one typically works on a hexagonal lattice (or other lattice with degree at most $3$).
If we instead consider a theory of loops on a square lattice, a ``figure-$8$" loop configuration is possible for which it is ambiguous whether we have a single loop which crosses itself or two loops which touch at a corner.  If the degree is at most $3$ then closed $1$-chains
unambiguously define a disconnected set of closed loops: a spin $\uparrow$ on a $1$-cell corresponds to loop being present on that $1$-cell and a
spin $\downarrow$ corresponds to loop being absent.

There are two distinct choices typically made for the double semion model Hamiltonian on a hexagonal lattice\cite{LW}.
Both have the same ground state subspace and both have an excitation gap.
In both, we again decompose $H=H_++H_\square$, where $H_+$ is the same as the toric code $H_+$ for $p=1,d=2$ and
$H_\square=\sum_{\text{$2$ cells $c$}} H_c$ for $H_c$ defined as follows.
In the first choice, one sets
\be
\label{Hcdef1}
H_c=\frac{1 + i^{n_{bdry}(c)} \prod_{\text{$1$ cells $f \in \partial c$}} X_{f}}{2},
\ee
where $n_{bdry}(c)$ is equal to the number of legs leaving the hexagon $c$ which are in the $\uparrow$ state.  Thus, $i^{n_{bdry}(c)}$ is equal to the product of
the operators
$\begin{pmatrix} 1 & 0 \\ 0 & i \end{pmatrix}$ around the legs leaving the hexagon.
With this choice, the terms  $H_c$ commute with the terms $H_e$ and the terms $H_c$ pairwise commute with each other
when restricted to the eigenspace of $H_+$ with vanishing eigenvalue.

The second choice is to project the term in Eq.~(\ref{Hcdef1}) into the zero eigenspace of the operators $H_e$ for all vertices $e$ in the hexagon $c$.
With this choice, the terms all commute with each other.

Consider a configuration of loops on the sphere.  The term $H_c$ can
induce a transition between two configurations differing by a surgery move:
\[\begin{xy}\ar@{-}@/_5pt/ (0,2.5);(5,2.5) \ar@{-}@/^5pt/ (0,-2.5);(5,-2.5)\end{xy} \rightarrow \begin{xy}\ar@{-}@/_5pt/ (-2.5,-2.5);(-2.5,2.5) \ar@{-}@/^5pt/ (2.5,-2.5);(2.5,2.5)\end{xy}\]
To be in the zero eigenspace of $H$, the two configurations related by this move must enter the wavefunction with
the opposite sign.
Given any configuration of closed loops on the sphere, indeed this transition will always change the number of loops by $\pm 1$ and hence
the sign rule for this transition is consistent with the ground state wavefunction.

However, on other manifolds, such a transition may not change the number of loops.
Consider, for example, a single closed loop $(1,1)$ on a torus, winding once around the meridian and once around the longitude.  This can transition to
a loop with different relative windings $(1,-1)$, as in the following figure where opposite sides of the square
are identified
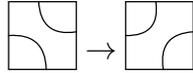
\[\begin{xy}\ar@{-}@/_5pt/ (2,6.5);(7,2.5) \ar@{-}@/^5pt/ (-2,2);(3,-2.5)
\ar@{-}@/5pt/ (-2,-2.5);(-2,6.5)
\ar@{-}@/5pt/ (7,6.5);(-2,6.5)
\ar@{-}@/5pt/ (7,6.5);(7,-2.5)
\ar@{-}@/5pt/ (7,-2.5);(-2,-2.5)
\end{xy} \rightarrow \begin{xy}\ar@{-}@/_5pt/ (-2,2);(2,6.5) \ar@{-}@/^5pt/ (3,-2.5);(7,2.5)
\ar@{-}@/5pt/ (-2,-2.5);(-2,6.5)
\ar@{-}@/5pt/ (7,6.5);(-2,6.5)
\ar@{-}@/5pt/ (7,6.5);(7,-2.5)
\ar@{-}@/5pt/ (7,-2.5);(-2,-2.5)
\end{xy}\]
Thus, the sign structure of the wavefunction on a torus obeys a more complicated rule.  In fact, consistency of the transition rule on
non-orientable surfaces can reduce the ground state degeneracy of the double semion model; see the next section.

\section{Some pathologies of $\DS$ and its potential generalizations}\label{sec:3}
For an arbitrary manifold $\Sigma$, let $TC(\Sigma)$ and $DS(\Sigma)$ denote the subspace of zero energy states.  In this section we consider
the degeneracy of this subspace for various manifolds.

%%%%% Insert - Degeneracy of DS on non-orientable surface %%%%%
Closed non-orientable surfaces are diffeomorphic to: $P$, $2P$, $3P$, $\ldots$, where $P$ is the projective plane, $RP^2$, and the positive integer represents that number of $P$s tubed (``connected sum'') together.

\begin{theorem}\label{thm:A}
\[\dim(\DS(tP)) = 2^{t-1}\]
\end{theorem}

\subsubsection*{Note}

For orientable surfaces, $\DS$ and $\TC$ have equal dimension, but for closed non-orientable surfaces, $\frac{\dim(\DS(tP))}{\dim(\TC(tP))} = \frac{1}{2}$.

We will prove theorem \ref{thm:A} later in this section.
First we collect a few elementary topological facts and then prove some special cases as a warmup before proving the general result.

\begin{fact}\label{fact:1}
If a closed surface $\Sigma^2 = \partial M^2$ is the boundary of a $3$-manifold,
\[\dim\ker\left(H_1\left(\Sigma,Z_2\right)\rightarrow H_1\left(M;Z_2\right)\right) = \frac{1}{2}\dim H_1\left(\Sigma; Z_2\right).\]
In particular, $\dim H_1\left(\Sigma;Z_2\right)$ is even.  Equivalently the Euler characteristic $\chi(\Sigma)$ is even.

More generally, if $N^{2k}$ is the boundary of a ($2k+1$)-dimensional manifold $M^{2k+1}$, $\geq 1$, then:
\[\dim\ker\left(H_k\left(N;Z_2\right)\right)\rightarrow H_k\left(M;Z_2\right) = \frac{1}{2}\dim H_k\left(N;Z_2\right)\]
Again, this implies $\dim H_k\left(N, Z_2\right)$ is even, and the Euler characteristic $\chi(N)$ (base of Betti numbers with any coefficients) is also even.

Since we only use $Z_2$-coefficients, they are now dropped from the notation.
\end{fact}
\begin{proof}
Consider the long exact sequences of the pair $(M,N)$ in both cohomology and homology with the vertical isomorphism given by Poincar\'{e} and Lefschetz dualities:
\[\xymatrix{
H^k(M,\partial) \ar[r]\ar[d]^{\text{L.D.}} & H^k(M) \ar[r]\ar[d]^{\text{L.D.}} & H^k(\Sigma) \ar[r]\ar[d]^{\text{P.D.}} & H^{k+1}(M,\partial) \ar[r]^{\hat{\alpha}}\ar[d]^{\text{L.D.}} & H^{k+1}(M) \ar[d]^{\text{L.D.}} \\
H_{k+1}(M) \ar[r]^{\alpha} & H_{k+1}(M,\partial) \ar[r]^{\beta} & H_{k+1}(\Sigma) \ar[r]^{i_\ast} & H_k(M) \ar[r]^{\gamma} & H_k(M,\partial)
}\]

The universal coefficient theorem (with field coefficients $F$) provides a natural isomorphism $H^\ast(\; ; F) \overset{\cong}{\longrightarrow} \Hom\left(H_\ast(\; ; F), F\right)$ implying that $\alpha$ and $\hat{\alpha}$ above are $\Hom$-dual.  Thus $\Hom(\coker \alpha, Z_2) = \ker\hat{\alpha}$.  In particular, $\dim(\coker\alpha) = \dim\left(\ker\hat{\alpha}\right) = \dim(\ker\gamma)$, the later from the commutative square.  By exactness of the homology sequence, $\coker\alpha\cong\image\beta = \ker i_\ast$.  Thus $\dim\left(\ker i_\ast\right) = \dim(\ker \gamma)$.  But using exactness once more, $\image i_\ast = \ker\gamma$ so we conclude: $\dim(\ker i_\ast) = \dim(\image i_\ast)$.  Since $\dim H_k(\Sigma) = \dim(\ker i_\ast) + \dim(\image i_\ast)$, the fact is proved.
\end{proof}

\begin{corollary}
$tP$, $t$ odd, do not bound any compact $3$-manifold. For $t$ even, $tP$ does bound a compact $3$-manifold.
\end{corollary}
\begin{proof}
An elementary calculation shows $H_1(tP)\cong Z_2^t$ and thus has odd dimension for $t$ odd.

Let $K = 2P$ be the Klein bottle.  $K$ is a $S^1$-bundle over $S^1$ with reflection monodromy.  ($S^1$ is the circle.)  This monodromy clearly extends to a reflection of the disk $D^2$, so $K$ bounds a disk bundle over $S^1$.  This handles the case $t=2$.  In general, for $t = 2j$ one may tube together $j$ copies of this disk bundle to produce the desired $3$-manifold.
\end{proof}

\begin{fact}\label{fact:3}
Any two embedded codimension $=1$ submanifolds $P$ and $Q$ of a manifold $N^d$ which represent the same $Z_2$-homology class $\alpha\in H_{d-1}\left(N;Z_2\right)$ are cobordant: There is a compact $W^d$ with $\partial W^d = P\coprod Q$ (and in fact $W$ also maps into $N$ extending the inclusions of $P$ and $Q$).

We will apply this where $P$ and $Q$ are closed surfaces and $N$ is a closed $3$-manifold.
\end{fact}
\begin{proof}
$H_{d-1}\left(N,Z_2\right)\cong H^1\left(N,\partial N,Z_2\right)$ by Lefschetz duality.  $H^1$ is \emph{represented} as homotopy classes of maps $\left[N,RP^\infty\right]$ (with a relativization on $\partial N$ to the base point if $\partial N\neq\emptyset$).  We use $\partial X$ to denote the boundary of any space $X$.  The classifying space for $H^1\left(\;\;;Z_2\right)$ is $K\left(Z_2,1\right)$, the space with $\pi_1\cong Z_2$ and all higher homotopy groups vanishing; it is the infinite projective space $RP^\infty$.  Its characteristic class $\iota \in H^1\left(RP^\infty; Z_2\right)$ is dually represented by a codimension one $RP^{\infty-1}\subset RP^\infty$.  A class $p\in H^1\left(N, \partial N; Z_2\right)$ is represented by a map $f_p : (N, \partial N) \rightarrow \left(RP^\infty, \ast\right)$; $\operatorname{P.D.}(p) = \left[f_p^{-1}(RP)^{\infty-1}\right]$.  The submanifolds $P$ and $Q$ arise as different but homotopic ways of making the representing map transverse to $RP^{\infty - 1} \subset RP^\infty$.  $W$ is produced by a \emph{relative} application of transversality to the homotopy, $F$.  That is, $P = f_p^{-1}\left(RP^{\infty-1}\right)$, $Q = f_q^{-1}\left(RP^{\infty-1}\right)$, and $W = F^{-1}\left(RP^{\infty-1}\right)$.
\end{proof}

Now consider $\DS\left(RP^2\right)$.  We first consider the question: Is there any ground state wave function (g.s.w.f.) $\psi(\;)\in\DS\left(RP^2\right)$ with nonzero weight on the empty picture or \emph{blank} multi-loop diagram.  The answer is no, since starting with the empty picture, a loop may be swept over $P$ with $3$ local events, each of which produces a $-1$ phase difference between the multi-loops before and after each event.  The product of these three signs is $-1$ showing $\psi(\emptyset) = -\psi(\emptyset) = 0$.
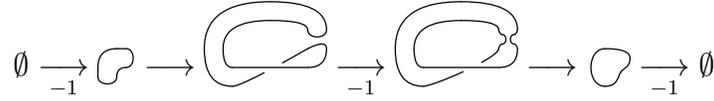
\[\emptyset\underset{-1}{\longrightarrow}
  \begin{xy}<5mm,0mm>:
  (0,-0.25);(0.25,0)**\crv{(0,0)} % inside
  ,(0,-0.25);(-0.5,-0.25)**\crv{(0,-0.5)&(-0.5,-0.5)} % bottom
  ,(0.25,0);(0.25,0.5)**\crv{(0.5,0)&(0.5,0.5)} % right
  ,(0,0.5);(-0.5,0)**\crv{(-0.5,0.5)} % outside
  ,(0,0.5);(0.25,0.5)**\dir{-} % outside to right
  ,(-0.5,0);(-0.5,-0.25)**\dir{-} % outside to bottom
  \end{xy}
  \longrightarrow
  \begin{xy}<5mm,0mm>:
  (-1,0);(1,0)**\dir{-} % horizontal cross
  ,(1,0);(1.5,0.5)**\crv{(1.5,0)} % bottom cap bottom
  ,(1,0.5);(1.5,0.5)**\crv{(1.5,0.75)} % bottom cap top
  ,(-1,0);(-1.25,0.5)**\crv{(-1.25,0)} % inside left bottom
  ,(-1.25,0.5);(0,1.25)**\crv{(-1.25,1.25)} % inside left top
  ,(0,1.25);(1,1)**\crv{(1,1.25)} % inside top right
  ,(1,1);(1.5,1)**\crv{(1,0.75)&(1.5,0.75)} % top cap
  ,(1.5,1);(0,1.75)**\crv{(1.5,1.75)} % outside top right
  ,(0,1.75);(-1.75,0.5)**\crv{(-1.75,1.75)} % outside left top
  ,(-1.75,0.5);(-1,-0.5)**\crv{(-1.75,-0.5)} % outside left bottom
  \ar@{-}@`{(-0.75,-0.5)}|(0.7)\hole (-1,-0.5);(1,0.5) % diagonal cross
  \end{xy}
  \underset{-1}\longrightarrow
  \begin{xy}<5mm,0mm>:
  (-1,0);(1,0)**\dir{-} % horizontal cross
  ,(1,0);(1.5,0.5)**\crv{(1.5,0)} % bottom cap bottom
  ,(-1,0);(-1.25,0.5)**\crv{(-1.25,0)} % inside left bottom
  ,(-1.25,0.5);(0,1.25)**\crv{(-1.25,1.25)} % insde left top
  ,(0,1.25);(1,1)**\crv{(1,1.25)} % inside top right
  ,(1.5,1);(0,1.75)**\crv{(1.5,1.75)} % outside top right
  ,(0,1.75);(-1.75,0.5)**\crv{(-1.75,1.75)} % outside left top
  ,(-1.75,0.5);(-1,-0.5)**\crv{(-1.75,-0.5)} % outside left bottom
  ,(1,0.5);(1.09375,0.625)**\crv{(1.03,0.625)} % merge left bottom
  ,(1.0975,0.625);(1.0975,0.875)**\crv{(1.2225,0.625)&(1.2225,0.875)} % merge left middle
  ,(1.09375,0.875);(1,1)**\crv{(1.03,0.875)} % merge left top
  ,(1.5,0.5);(1.40625,0.625)**\crv{(1.47,0.625)} % merge right bottom
  ,(1.40625,0.625);(1.40625,0.875)**\crv{(1.28125,0.625)&(1.28125,0.875)} % merge right middle
  ,(1.40625,0.875);(1.5,1)**\crv{(1.47,0.875)} % merge right top
  \ar@{-}@`{(-0.75,-0.5),(1,0.25)}|(0.575)\hole (-1,-0.5);(1,0.5) % diagonal cross
  \end{xy}
  \longrightarrow
  \begin{xy}<5mm,0mm>:
  (0.165,0.25);(-0.5,-0.5)**\crv{(0.165,0)&(-0.165,0)&(-0.165,-0.5)} % inside
  ,(-0.5,0.5);(-0.5,-0.5)**\crv{(-1,0.5)&(-1,-0.5)} % left
  ,(-0.5,0.5);(0.165,0.25)**\crv{(0.165,0.5)} % top
  \end{xy}
  \underset{-1}\longrightarrow\emptyset
\]

In words, we create a small circle in a M\"{o}bius band $\mathcal{M}$, expand it around the band, and recouple it so that it becomes parallel to the boundary of $\mathcal{M}$.  Since $P = \mathcal{M}\cup_\partial\text{disk}$, the recoupled loop now shrinks to a point across the disk and disappears.  The sign is $(-1)^{\chi(P)}$ as each ``event'' corresponds to a critical point of a Morse function on $P$.  In the simplest case there are three.  The crossing in the above figure represents the half-twist in the M\"{o}bius band; the fourth image (out of six images in the sequence) represents the loop parallel to the boundary of the band.  This boundary is drawn, after projecting the band into the plane, as a loop with a single self-crossing.

Similar argument shows that for any $\DS$ g.s.w.f.~$\psi_{tP}(\text{empty}) = 0$, $t$ odd.  The sweepout now has $t+2$ events: the birth of a circle, $t$ reconnections, and finally the death of the circle.  We have shown:
\begin{lemma}
For $t$ odd, any $\DS$ ground state wave function satisfies $\psi_{tP}(\text{empty}) = 0$.
\end{lemma}\qed

To gain some intuition, let us now consider a g.s.w.f.~$\psi_P(\;)$ with $\psi_P\left(RP^1\right) = 1$, which evaluates to one on the essential loop.  Consider $P$ as the unit disk $D$ modulo antipodal boundary identifications.  There is an embedding of the Klein bottle $K \hookrightarrow P\times S^1$ corresponding to revolving $RP^1$ (a diameter of $D$) by $\pi$ as one traverses the $S^1$ coordinate.

\begin{figure}[hbpt]
\[\begin{xy}<5mm,0mm>:
(0,0)*{\bullet}*\xycircle<71pt>{-}
,(0,1)*{RP^1}
,(-5,0);(5,0)**\dir{-}
,(0,-2.5)*{P}
,(10,0)*{\text{yields $K\hookrightarrow P\times S^1$}}
\ar@`{(2.25,3.5),(-2.25,3.5)}_{\pi\text{-rotation}} (2.5,0.5);(-2.5,0.5)
\end{xy}\]
\caption{}
\label{fig:1}
\end{figure}

Any $1$-parameter history of multi-loops including births, deaths, and reconnections beginning and ending with $RP^1$ defines a closed surface $S\hookrightarrow P\times S^1$.  The second homology of $P\times S^1$ obeys the K\"{u}nneth formula:
\[H_2\left(P\times S^1\right)\cong H_1(P)\times H_1\left(S^1\right)\oplus H_2(P) \times H_0\left(S^1\right) \cong Z_2\oplus Z_2\]
Choosing the initial multi-loop to be $RP^1$ or any other choice in that homology class fixes the first factor to be $1\in Z_2$.  The two remaining choices corresponding to $H_2(P)\cong Z_2$ are exemplified by the product history (nothing happens!) where $S\cong S^1\times S^1 := T$, a torus, and the twisting history (Figure \ref{fig:1}) where $S\cong K$.

Any other choice of history $S^\prime$ is homologous in $P\times S^1$ to either $T$ or $K$.  Since both $T$ and $K$ are cobordant to $\emptyset$ (bound closed manifolds), and (by Fact \ref{fact:3}) $S^\prime$ is cobordant to $T$ or $K$, we see that $S^\prime$ is cobordant to $\emptyset$.  (Concretely, glue the cobordism from $S^\prime$ to $T$ or $K$ to the null cobordism along their common boundary ($T$ or $K$).)  Thus by Fact \ref{fact:1}, $\chi\left(S^\prime\right) = \text{even}$.  Consequently the skein relations between multi-loops in the non-trivial homology class are consistent: an even number of $(-1)$s appear as we move from any essential multi-loop configuration through others and finally back to itself.  This means that there is a nonzero $\DS$ g.s.w.f., with $\psi_P\left(RP^1\right) = 1$.  With this warmup we can complete the proof of the theorem.

\begin{proof}[Proof of \ref{thm:A}]
Let $\Sigma$ be a closed surface with first $Z_2$-Betti number $=b$.  Let the elements $e_1,\ldots,e_{2^b}$ of $H_1(\Sigma)$ span $\mathbb{C}^{2^b}$.  There is a map $\theta,\theta(\psi) = \left(\psi\left(E_1\right),\ldots,\psi\left(E_{2^b}\right)\right)$ which injects the Hilbert space $\DS(\Sigma)$ into $\mathbb{C}^{2^b}$, where $E_i$ is any multi-loop representing $e_i$.

To determine $\dim(\DS(\Sigma))$, it is sufficient to characterize the $e_i$ for which $\psi\left(E_i\right)$ can be nonzero.

\begin{claim}\label{claim}
There is a $\DS$ g.s.w.f.~$\psi_i$ so that $\psi_i\left(E_i\right) = 1$ and $\psi_i\left(E_j\right) = 0$, $i\neq j$, if and only if $w_1\left(e_i\right) + \chi(\Sigma) = 0$.  $w_1$ is the first Stiefel-Whitney class. $\left\{\theta\left(\psi_i\right)\right\}$ spans $\image\theta$.

If $\Sigma$ is orientable, note that $w_1\left(e_i\right) + \chi(S) \equiv 0+0 = 0\mod 2$.

Also note the claim implies the theorem since, for a non-orientable $\Sigma = tP$, $w_1 : H_1(\Sigma)\rightarrow Z_2$ is onto, so exactly $\frac{1}{2}$ of the $2^t$ first homology classes have $w_1\left(e_i\right) = 0$ and exactly $\frac{1}{2}$ have $w_1\left(e_i\right) = 1$.
\end{claim}
\begin{proof}[Proof of \ref{claim}]
As in our warmup discussion of $RP^2$, there are precisely two homology classes $E_i\times S^1$ in $\Sigma\times S^1$ which realize a given $e_i$ upon intersection with $\Sigma\times 1$.  One is represented by the constant history.  Any surface $S$ (history) in this class is cobordant to $E_i\times S^1$ which bounds $E_i\times D^2$, so $S$ is cobordant to $\emptyset$, hence $\chi(S) \equiv 0\mod 2$.  In this case, the skein relations merely say that $\psi\left(E_i\right) = (-1)^\text{even}\psi\left(E_i\right)$ and pose no restriction.  The other possibility is more interesting: $\Sigma$ homologous to $E_i\times S^1 + \Sigma\times 1$.  In this case we compute:
\be\label{eqn:ast}
\chi(S) \equiv w_1\left(E_i\right) + \chi(\Sigma)\mod 2.\tag{$\ast$}
\ee
$\psi\left(E_i\right)$ is forced to vanish exactly when the right-hand side is odd.  Thus it suffices to establish (\ref{eqn:ast}) for any representative of the cobordism class of $S$.  There is no loss of generality taking $E_i$ to be a single circle, which we do.  The natural choice is a ``resolution'' of the union $\left(E_i\times S^1\right)\cup (\Sigma\times 1)$, a union of a torus and a copy of $S$ along $E_i\times 1$.  This means that in each normal disk cross-section of $E_i\times 1$ we resolve the crossing
$\begin{xy}<5mm,0mm>:
(-0.5,0);(0.5,0)**\dir{-}
,(0,-0.5);(0,0.5)**\dir{-}
\end{xy}$
of the two surfaces as either
$\begin{xy}<5mm,0mm>:
(-0.5,0);(0,0.5)**\crv{(0,0)},(0,-0.5);(0.5,0)**\crv{(0,0)}
\end{xy}$
or
$\begin{xy}<5mm,0mm>:
(-0.5,0);(0,-0.5)**\crv{(0,0)},(0,0.5);(0.5,0)**\crv{(0,0)}
\end{xy}$
.  We try to do this continuously all the way around $E_i\times 1$, and we will succeed precisely when $w_1\left(e_i\right) = 0$, for in this case $E_i\subset \Sigma$ has a cylindrical neighborhood (as does $E_i\times 1\subset E_i\times S^1$).  We will fail precisely when $w_1\left(E_i\right) = 1$, i.e., $E_i$ has a M\"{o}bius band neighborhood in $\Sigma$.  In this case, near the contradictory point we see the two resolutions fitting together to form a ``baseball'' curve near the final point (see Figure \ref{fig:2}).

\begin{figure}[hbpt]
\[\begin{xy}<5mm,0mm>:
(0,0)*\xycircle<71pt>{-} % main circle
,(0,0)*\xycircle<21pt,71pt>{++\dir{--}} % vertical equator
,(1.5,0);(5,0)**\crv{(2.5,0)&(5,1)} % right middle
,(5,0);(0,-5)**\crv{(4,-4)} % right bottom
,(0,-5);(-5,0)**\crv{(-0.25,-5)&(-0.125,-4)&(-3.5,-3.5)&(-4.5,0.5)&(-5,0.5)} % left bottom
,(-5,0);(4.33,2.5)**\crv{~**\dir{.}(-5,-0.5)&(-4.5,-0.5)&(4.5,1.25)} % dotted
,(4.33,2.5);(0,5)**\crv{(4.2,2.75)&(3.25,3.75)&(1,4.5)&(1,5)} % top left
,(0,5);(1.5,0)**\crv{(-3,4)&(-3,0)} % left middle
\end{xy}\]
\caption{}
\label{fig:2}
\end{figure}

Thus, to compute the resolution, a disk $\delta$ must be added bounding the baseball curve in the ball in Figure \ref{fig:2}.  Using the additivity formula for Euler characteristic
\be\label{eqn:astast}
\chi(A\cup B) = \chi(A) + \chi(B) - \chi(A\cap B)\tag{$\ast\ast$},
\ee
one readily checks (\ref{eqn:ast}) by showing that when $w_1\left(E_i\right) = 0$, $\chi(\text{Res}) = \chi(\Sigma)$ and when $w_1\left(E_i\right) = 1$, the partial gluing along two arcs (not circles) contributes $-2$ to $\chi$ while $\delta$ contributes $+1$, so in total $\chi(\text{Res}) = \chi(\Sigma) - 1$.
\end{proof}
\end{proof}

\subsubsection*{Note}
As an example, on $K$ any $\DS$ g.s.w.f.~$\psi$ evaluates to zero on the two non-orientable $1$-cycles, but is arbitrary on the oriented cycles.  In general the restriction on a $\DS$ g.s.w.f. on $tP$ is that it evaluate to zero on non-orientable cycles if and only if $t$ is even and that it evaluate to zero on orientable cycles if $t$ is odd.
%%%%% End insert %%%%%

In subsequent sections we will identify classes of $d$-manifolds $X$ for which $\GDS$s on $X$ have ground state degeneracy (gsd) $= \left|H_1\left(X; Z_2\right)\right|$.  In the remainder of this section we show a dramatic drop in ground state degeneracy if $\DS$ is generalized to a theory of multi-loops of dimension one fluctuating within closed $3$-manifolds.  This motivates our focus on fluctuating codimension one surfaces within $d$-dimensional manifolds.

First consider a theory in which \emph{unframed} multi-loops fluctuate in a $3$-manifold $M^3$ with the zero mode ground state wave function (zgswf) experiencing a phase $=-1$ with each Morse transition.  Even locally, within a $3$-ball, a closed cycle with one Morse transition exists:

\begin{equation}\label{line:3.1}
\begin{xy}<5mm,0mm>:
(-0.67,0)*\xycircle<10pt>{}
,(4.25,0);(3.75,0.3125)**\crv{(4.5,-0.25)&(5.875,0.25)&(5,1.25)&(5.5,1.5625)&(6.5,1.375)&(7.25,1)&(6.75,-0.75)&(2.5,-0.75)&(5.5,1)&(4.5,2)&(2.875,1.25)}
,(11.375,0);(10.875,0.3125)**\crv{(11.625,-0.25)&(13,0.25)&(12.625,1.5625)&(9.625,-0.75)&(13.875,-0.75)&(14.375,1)&(13.625,1.375)&(12.625,1)&(11.625,2)&(10.5,1.25)}
,(18.5,0)*\xycircle<10pt>{}
\ar^(0.33){\text{isotopy}}@`{(1.33,0.5),(2,-0.5)} (0.67,0);(2.67,0)
\ar^{\text{Morse}} (7.67,0);(9.67,0)
\ar^(0.33){\text{isotopy}}@`{(15.67,0.5),(16.33,-0.5)} (15,0);(17,0)
\end{xy}
\end{equation}

The local inconsistency: $\psi(\text{circle}) = -\psi(\text{circle})$ forces any zgswf to be identically zero on any $M^3$.

The explanation for this local inconsistency is that we neglected a normal framing, which we now add.  So we now assume all Morse transitions respect framings.  The beginning and end pictures on line (\ref{line:3.1}) now have different framings so there is no inconsistency.

However, even in the context of framed multi-loops, we find that the Hilbert space $V(M)$ of this theory for any closed oriented $3$-manifold $M$ has dimension $=1$, and below (Theorem \ref{thm:3.6}) characterize the single nontrivial sector.

\subsubsection*{Note}
It is proved in \cite{WW} that for any modular tensor category (MTC), the skein space on a closed oriented $3$-manifold has dimension one, so the dimension count in Theorem \ref{thm:3.6} is merely a special case, since ``semions'' form a MTC.

Given a closed $3$-manifold, let $E\in H_1(M)\overset{\text{P.D.}}{\cong}H^2(M)\cong\hom\left(H_2(M);Z_2\right)$ denote the class which assigns to any closed surface embedded in $M$, $\Sigma\hookrightarrow M$, its Euler class $\chi(\Sigma)\mod 2$.

\begin{theorem}\label{thm:3.6}
Let $M$ be a closed oriented $3$-manifold and let $V(M)$ be the ``Hilbert space'' of framed multi-loops embedded in $M$ with skein relation: ``Morse $=-1$''.  $\dim V(M)=1$, with $E$ being the nontrivial sector.
\end{theorem}

\begin{proof}
The argument closely parallels the proof of Theorem \ref{thm:A}, so it is only sketched.  Let $L$ be a loop or multi-loop in class $F\in H_1(M)$.  Consider a ``history'' $h$ beginning at $L\times 0\subset M\times 0\subset M\times [0,1]$ and ending at $L\times 1\subset M\times 1\subset M\times[0,1]$ and then identify boundaries to embed $M\times(0,1)\hookrightarrow M\times[0,1] / m\times 0 = m\times 1 = M\times S^1$.  Let us compute the mod $2$ Euler class of the closed history $\bar{h} \subset M\times S^1$.  $[\bar{h}] = [L] \times [S^1] + x\times [\text{pt.}]\in H_2\left(M\times S^1\right)$, where $x\in H_2(M)$.  Using the data of the normal framing, one computes:

\[\chi\left(\bar{h}\right)\equiv\chi\left(L\times S^1\right) + \chi(x) + \left|L\times S^1\cap x\right|\mod 2\]
\[\equiv 0 + \chi(x) + F\cdot x\mod 2\]

The right-hand side is zero for all histories, i.e., for all $x\in H_2(M)$, if and only if $F = [L] = E$, which is precisely the case that the sector survives in $V(M)$.
\end{proof}

\section{The theories we study: ``fluctuating ($d-1$) submanifold of $X^d$''}

A smooth closed $X^d$ manifold will be the home or ``ambient space'' of our generalized double semion ($\GDS$) theory.  We consider a fixed combinatorial structure $\mathcal{C}$, a generic cellulation, on $X$.  Associated to $\mathcal{C}$ is a ``crude Hilbert space'' $\mathcal{H}$ consisting of one qubit for each ($d-1$)-cell of $\mathcal{C}$, spanned by $\langle\uparrow = \text{present}, \downarrow = \text{absent}\rangle$.  There is a Hamiltonian:
\[H = H_+ + H_\square\]
acting on $\mathcal{H}$.  We further write:
\[H_+ = \sum_{d-2\text{ cells }e} H_e\text{, }H_\square = \sum_\text{$d$-cells $c$} H_c.\]
$H_+$ has a term for each ($d-2$)-cell $P$ which is zero if an even number of ($d-1$)-cells meeting $P$ are $\uparrow$ (present) and $1$ if an odd number are $\uparrow$.  This is the same as the $H_+$ term for the generalized toric code.

$H_\square$ has a term $H_c$ for each $d$-cell $c$.  This term is
\be\label{eq:4.1}
H_c=\frac{1-O_c}{2},
\ee
where
\be
\label{Ocdef}
O_c=\pm \prod_{\text{$(d-1)$ cells $f \in \partial c$}} X_{f}.
\ee
The operator $X_{f}$ is the Pauli-$X$ operator acting on the cell $f$, changing the state of the cell.  The sign in (\ref{Ocdef}) is $-1 \times -1^{\chi\left(\uparrow_c\right)}$ where $\chi$ is Euler characteristic and $\uparrow_c$ is the codimension $=0$ submanifold of $\partial c$ consisting of the union of ($d-1$) cells of $\partial c$ which are labeled $\uparrow$ in the state on which $H_c$ acts.  We will show in the appendix that this subset is indeed a smooth submanifold with corners.  This is the reason that we chose a generic cellulation; it is analogous to the reason for defining the double
semion model in two dimensions on a hexagonal lattice, rather than a square lattice.

The boundary of $\uparrow_c$ regarded as a submanifold of $\partial c$ is the same as the boundary of $\downarrow_c$ and is equal to $\uparrow_c \cap \downarrow_c$.
For even $d$, the Euler characteristic of $\uparrow_c \cap \downarrow_c$ is equal to twice the Euler characteristic of $\uparrow_c$, so
$-1^{\chi\left(\uparrow_c\right)}$ is equal to $i^{\chi\left(\uparrow_c \cap \downarrow_c \right)}$.  This makes the sign more reminiscent of the sign in the double semion model
in two dimensions, where there is a factor of $i$ for every leg leaving the hexagon.

\subsubsection*{Explanation of generic cellulations}

Generic cellulations may be defined as those divisions of a smooth closed $d$-manifold $X^d$ into a union of smooth $k$-cells, $0\leq k\leq d$, piecewise smooth on their boundaries, which obey the local combinatorics of the dual cells to a smooth triangulation of $X$.  For example, on a surface ($d=2$), a trivalent graph with contractible complementary regions determines a generic cellulation.  In 3D the point singularities are cone over the $1$-skeleton of a tetrahedron, etc.  A second construction which works in any dimension is to start with not a smooth triangulation but rather Riemannian metric.  Then any sufficiently dense set $S$ of points, if perturbed to be generic, will determine Voronoi cells (the $d$-cell interiors consist of all points closest to some $s\in S$) which define a generic cellulation.  We work with these structures rather than some fixed lattice for two reasons:
\begin{enumerate}
    \item every union $U$ of ($d-1$) cells whose $Z_2$-boundary vanishes ($\Leftrightarrow$ every ($d-2$) cell meets an even number of ($d-1$) cells of $U$) is a submanifold.
    \item heritability: The intersection pattern induced on the boundary of any cell of a generic cellulation is \emph{also} a generic cellulation in a lower dimension.  For more details, see appendix.
\end{enumerate}

\begin{lemma}
\label{termscommlemma}
All terms of $H_+$ pairwise commute.  All terms in $H_+$ commute with all terms in $H_\square$.  All terms in $H_\square$ pairwise commute when restricted to the
eigenspace of $H_+$ with vanishing eigenvalue.
\end{lemma}

\begin{proof}
The terms in $H_+$ pairwise commute since they are all diagonal in the product basis $|\uparrow\rangle,|\downarrow\rangle$.
We assume the reader is familiar with the analogous verification that the terms of Kitaev's \cite{K} toric code ($\TC$) commute.  As with $\TC$, adjacent $+$ and $\square$ operators share
two anticommuting factors and therefore commute.  The local sign is irrelevant.

Consider a pair of terms $H_{c_1},H_{c_2}$ in $H_\square$.
They trivially commute unless $\partial c_1$ and $\partial c_2$ share some $(d-1)$ cell $f$.
Consider a state $\Psi$ in the product basis on which $H_+$ vanishes.
We show that $O_{c_1} O_{c_2} \Psi=O_{c_2} O_{c_1} \Psi$.
The states $O_{c_1} O_{c_2} \Psi$ and $O_{c_2} O_{c_1} \Psi$ are both product states and have the same spin configuration.  However, we must check that they have the same sign.
Suppose first that $f$ is labeled $\downarrow$ in $\Psi$, so that $O_{c_2}$ changes the set of $\uparrow$ spins on $c_1$ to $\uparrow_{c_1} \cup f$.
So, the sign for $O_{c_1} O_{c_2}$ is equal to
$-1^{\chi\left(\uparrow_{c_1} \cup f \right)}        -1^{\chi\left(\uparrow_{c_2}\right)}$.
A similar computation for $O_{c_2} O_{c_1}$ gives
$-1^{\chi\left(\uparrow_{c_2}\cup f\right)}          -1^{\chi\left(\uparrow_{c_1}\right)}$.
So we must check that
\be\label{eq:4.2}
-1^{\chi\left(\uparrow_{c_1} \cup f \right)}        -1^{\chi\left(\uparrow_{c_2}\right)}=
-1^{\chi\left(\uparrow_{c_2}\cup f\right)}          -1^{\chi\left(\uparrow_{c_1}\right)}.
\ee
By additivity formula of Euler characteristic, $\chi\left(\uparrow_{c_1}\cup f\right)=\chi\left(\uparrow_{c_1}\right)+\chi(f)-\chi(\uparrow_{c_1} \cap f)$.
So, we must check that
\be
\chi(\uparrow_{c_1}\cap f)+\chi(\uparrow_{c_2}\cap f)=0 \mod 2.
\ee
However, by assumption that $H_+$ vanishes on $\Psi$, $\uparrow_{c_1}\cap f = \uparrow_{c_2}\cap f$.
If instead $f$ is labeled $\uparrow$ in $\psi$, we must check that
$-1^{\chi\left(\uparrow_{c_1} \setminus f \right)}        -1^{\chi\left(\uparrow_{c_2}\right)}=
-1^{\chi\left(\uparrow_{c_2}\setminus f\right)}          -1^{\chi\left(\uparrow_{c_1}\right)}$.
A similar calculation using additivity formula for Euler characteristic shows this case also.
\end{proof}

\begin{lemma}
All terms of $H$ are projectors.
\end{lemma}

\begin{proof}
This is obvious for terms of $H_+$.  For $H_c$ one needs to check that the operator $O_c$
is order two.
To do this, one needs to check that the signs agree:
\be\label{eqn:astastast}
(-1)^{\chi\uparrow_c} \equiv (-1)^{\chi\downarrow_c}\mod 2
\ee
where $\downarrow_c$ is the complement in $\partial c$ of $\uparrow_c$.  But this follows immediately from three facts:
\begin{enumerate}
    \item $\partial c$ is a ($d-1$)-sphere $S^{d-1}$, $\chi(\partial c) = \begin{cases} 0 & d\text{ even} \\ 2 & d\text{ odd} \end{cases}$,
    \item the additivity formula:
        \[\chi\left(\uparrow_c\right) + \chi\left(\downarrow_c\right) - \chi\left(\uparrow_c\cap\downarrow_c\right) = \chi(\partial c) \equiv 0\mod 2,\]
        and
    \item Since $\left(\uparrow_c\cap\downarrow_c\right)$ bounds (it in fact bounds both $\uparrow_c$ and $\downarrow_c$), $\chi\left(\uparrow_c\cap\downarrow_c\right)\equiv 0\mod 2$ since the Euler characteristic of any bounding manifold is even.
\end{enumerate}

The upshot is (\ref{eq:4.2}).
\end{proof}

As in the double semion case in two dimensions, it is possible to define a Hamiltonian which is a sum of commuting projectors by projecting the operators
$H_c$ into the zero eigenspace of the terms $H_e$ for $e$ a face of $c$, setting $H_\square=\sum_c H_c^{proj}$, with
\be
\label{Hcdef2}
H_c^{proj}=\Bigl( \prod_\text{$d-2$ cells $e$, s.t. $e$ is a face of $c$} 1-H_e \Bigr) H_c \Bigl( \prod_\text{$d-2$ cells $e$, s.t. $e$ is a face of $c$} 1-H_e \Bigr).
\ee
Then,
\begin{lemma}
All terms $H_c^{proj}$ pairwise commute and commute with all terms $H_e$.
\end{lemma}
\begin{proof}
The proof is the same as lemma \ref{termscommlemma}, noting that that proof only used the vanishing of $H_e$ for $H_e$ a face of $c_1,c_2$.
\end{proof}

\subsubsection*{Explanation of the sign of $H_c$}

We should think of $H_c$ as inducing fluctuations between codimension one submanifolds.  The most basic possible fluctuations are Morse transitions which near the transition point, by definition assume the following local form, for index $0\leq k\leq d-1$:
\[-x_1^2 - x_2^2 -\cdots -x_k^2 + x_{k+1}^2 +\cdots + x_{d-1}^2 = -\epsilon \Rightarrow\]
\[-x_1^2 - x_2^2 -\cdots - x_k^2 + x_{k+1}^2 +\cdots + x_{d-1}^2 = \epsilon\text{, }\epsilon > 0.\]
For example, for the usual double semions these are the transitions:
\[\emptyset \rightarrow \begin{xy}(0,0)*\xycircle<5pt>{}\end{xy}\text{ ($k=0$)}\]
\[\begin{xy}\ar@{-}@/_5pt/ (0,2.5);(5,2.5) \ar@{-}@/^5pt/ (0,-2.5);(5,-2.5)\end{xy} \rightarrow \begin{xy}\ar@{-}@/_5pt/ (-2.5,-2.5);(-2.5,2.5) \ar@{-}@/^5pt/ (2.5,-2.5);(2.5,2.5)\end{xy}\text{ ($k=1$)}\]
\[\begin{xy}(0,0)*\xycircle<5pt>{}\end{xy} \rightarrow \emptyset\text{ ($k=2$)}\]
Form the \emph{usual} Skein relations, each of which is associated with a phase $=-1$.  We would like to generalize this to higher dimensions by declaring a phase $=-1$ for each Morse transition.

Unfortunately, the transition $t := \uparrow_c\rightarrow\downarrow_c$ is not in general Morse but in general much more complicated.  However, a simple argument (below) enables us to write $t$ as a composition of $\left(\chi\left(\uparrow_c\right) + 1\right)_{\text{mod }2}$ Morse transitions.  This decomposition of $t$ is not unique but the parity of the number of Morse factors is.

Consider a Morse function: $f : c\rightarrow [0,1]$ with $f^{-1}(0) = \left(\uparrow_c\right)^-$, $f^{-1}(1) = \left(\downarrow_c\right)^-$, the $^{(-)}$ indicating deletion of a thin collar.  Let $f$ be free of critical points on $\partial c\setminus \left(\left(\uparrow_c\right)^-\cup\left(\downarrow_c\right)^-\right)$.

As we cross a critical level $s$ of $f$, $\chi\left(f^{-1}\left(0,s_+\right]\right)$ changes by $\pm 1$ according to $\operatorname{index}(s) =$ odd or even, respectively.  Since $\chi(c) = \chi\left(f^{-1}[0,1]\right) = 1$, the number (mod $2$) of Morse transitions of $f$ is $\chi\left(\uparrow_c\right)+1 = \chi\left(f^{-1}(0)\right)+1$.

\begin{observation}[``Morse $=-1$''] The Hamiltonian $H$ enforces on any zgswf $\psi$, the skein rule Morse $=-1$.  That is, if $E$ and $E^\prime$ are cellular $d-1$ submanifolds of $M$ which differ by $k$ Morse transitions, $\psi(E) = (-1)^k\psi\left(E^\prime\right)$.
\end{observation}

\section{Some properties of $\GDS$}

Let $X^d$ be the ambient manifold of dimension $d$.  As before, our model is regularized by choosing a generic cellulation $\mathcal{C}$ of $X$.  In this section and Section \ref{sec:6}, cutting and gluing constructions are used whose motivation is from the continuum.  It may be necessary to subdivide $\mathcal{C}$ to carry out these constructions.  Thus all statements should be interpreted as holding for some refinement $\mathcal{C}^\prime$ of $\mathcal{C}$.

The first property we observe regards the form of zero energy ground state wave functions (zgswfs).  Because of the term $H_+$, all zgswfs are superpositions of (cellular) submanifolds (of dimension $d-1$, and not necessarily oriented).  Because $H_\square$ does \emph{not} fluctuate between different classes of $H_{d-1}\left(X;Z_2\right)$, all zgswfs are superposition of zgswfs each with support confined to a fixed class $x\in H_{d-1}\left(X,Z_2\right)$.  The next lemma implies that there is at most one zgswf up to phase for any class $x\in H_{d-1}\left(X,Z_2\right)$.  However, as we saw when $d=2$ on non-orientable surfaces, it is not necessarily the case that \emph{each} possible sector $x\in H_{d-1}\left(X;Z_2\right)$ actually admits a zgswf.

\begin{lemma}
Let $\alpha$ and $\beta$ be two $\mathcal{C}$-cellular ($d-1$)-$Z_2$-cycles belonging to the same class $x\in H_{d-1}\left(X; Z_2\right)$.  There is a finite set $\{c_1,\ldots,c_j\}$ of $d$-cells of $\mathcal{C}$ so that $\alpha - \beta = \partial\left(c_1\cup\cdots\cup c_j\right)$.
\end{lemma}

\begin{proof}
This is directly from the definition of cellular homology.  No multiplicities or orientation need be considered since the coefficients are $Z_2$.
\end{proof}

\begin{corollary}
Let $\mathbb{C}^{2^b}$ be the Hilbert space spanned by the elements $e_1,\ldots,e_{2^b}$ of $H_{d-1}(X)$, and choose representative cellular cycles $E_i$ for $e_i$, $1\leq i\leq 2^b$.  The map $\theta$, $\theta(\psi) = \left(\psi\left(E_1\right),\ldots,\psi\left(E_{2b}\right)\right)$, is an injection of the zero-modes $\GDS(X)\rightarrow\mathbb{C}^{2^b}$.\qed
\end{corollary}

In Section \ref{sec:3} we discussed ``missing sectors'' when $X$ is a non-orientable surface.  In higher dimensions, there may be ``missing sectors,'' $\theta$ has cokernel, even when $X$ is orientable.  An extreme example is $X = CP^2$.

\subsubsection*{Example}

Let the ambient manifold $X$ be $\mathbb{C}P^2$, complex projective $2$-space, $d = \dim\left(\mathbb{C}P^2\right) = 4$.  $H_3\left(CP^2,Z_2\right)\cong 0$, so there is only one possible sector of zgswfs in $\GDS$, the trivial sector.  But, in fact, $\GDS\left(\mathbb{C}P^2\right)\cong 0$.  Just like $RP^2$, $\mathbb{C}P^2$ has a Morse function with three critical points
\[\left(Z_1,Z_2,Z_3\right) \rightarrow \frac{|Z_2|^2+2|Z_3|^2}{|Z_1|^2+|Z_2|^2+|Z_3|^2}.\]
Thus, as in Section \ref{sec:3}, $\psi(\emptyset) = -\psi(\emptyset) = 0$, so $\GDS\left(\mathbb{C}P^2\right)\cong 0$.  Of course $H$ will have \emph{some} least energy state but we do not know that it has topological significance.  We also do not know its entanglement properties.
If we work with a Hamiltonian that is a sum of commuting projectors using terms $H_c^{proj}$, then
the ground state is an eigenstate of every term separately.  Hence,
this ground state will be a zero energy eigenstate of the Hamiltonian on some ``punctured cellulation", obtained by removing the terms $H_c,H_e$
which have expectation value $+1$.
We now give precise conditions under which $\theta$ is an isomorphism:

\begin{theorem}\label{thm:5.3}
If $d$ is odd, $\theta$ is an isomorphism, $\GDS(X)\cong \mathbb{C}^{2^b}$.  Moreover, in each sector $x\in H_{d-1}(X)$, an explicit zgswf is given as $\psi_x(E) = i^{\chi(E)}$.

For d even and $x\in H_{d-1}(X)$, there is an zgswf $\psi_x(E)\neq 0$, $[E] = x\in H_{d-1}(X)$, if and only if $\chi(X) + \epsilon(x) \equiv 0\mod 2$, where $\epsilon(x)$ is defined below.

In general, there is no local formula\footnote{Even when $d=2$, it is only true on the $2$-sphere $S^2$ that the signs zgswfs are $(-1)^{\sharp(\text{components})}$.  Recall the example on the torus $T^2$, where the curves $(1,1)$ and $(1,-1)$ are in the same $Z_2$-sector, but $\psi((1,1)) = -\psi((1,-1))$ for any zgswf $\psi$.} for the zgswfs; however, if $X$ is a $Z_2$-homology $d$-sphere,\footnote{It is sufficient to assume $H_{d-1}(X) = 0 = H_\frac{d}{2}(X)$.} the formula for the sign of $\psi_x$ on $E$ is $\pm = (-1)^{s(E)}$, $s(E)$ also defined below.
\end{theorem}

\begin{proof}
Consider $d$ odd.  In this parity, the codimension one cellular submanifold $E$ is even dimensional.  Morse transitions of even dimensional manifolds always change $\chi$ by $\pm 2$.  The transition removes an $S^j\times D^{d-1-j}$ and replaces it with $D^{j+1}\times S^{d-j-2}$.
\[\chi\left(S^p\right) = \begin{cases} 0 & p\text{ odd} \\ 2 & p\text{ even} \end{cases}, \text{so}\]
\[\chi\left(S^j\times D^{d-1-j}\right) = \chi\left(S^j\right) \equiv_2 \chi\left(S^{d-j-2}\right) = \chi\left(S^{d-j-2}\times D^{j+1}\right).\]

The claim now follows from the gluing formula for Euler characteristic:

\be\label{eqn:5.4}
\chi(A\cup B) = \chi(A) + \chi(B) - \chi(A\cap B)
\ee

Thus every Morse transition $M\rightarrow N$ changes $i^{\chi(M)}$ by a factor of $-1$ (which is consistent with the fluctuation of $H_\square$).

Now consider $d =$ even.  Let $E$ be a cellular hypersurface of $X$ representing $x$.  Let $\nu$ be the normal bundle of $E$ in $X$ and $L\subset E$ be a $d-2$ dimensional embedded submanifold Poincar\'{e} dual to the first Stiefel-Whitney class $w_1(\nu) \in H^1(E)$, of the tangent bundle to $E$.

\begin{definition*}
$\epsilon(x) = \chi(L) \mod 2$.
\end{definition*}

\begin{lemma}
$\epsilon(x)$ is well-defined.
\end{lemma}

\begin{proof}
According to (\ref{fact:3}), any other hypersurface $E^\prime$ representing $x$ is bordant to $E$ via $V\subset X\times I$; $\partial_+ V = E^\prime$, $\partial_- V = E$.  Furthermore, denoting the normal bundle to $E^\prime$ (in $X\times 1$) by $\nu^\prime$, for any $L^\prime$ P.D. to $w_1\left(\nu^\prime\right)$ in $E^\prime$ there is a bordism (it is the P.D. to $w_1\left(\nu_{V\hookrightarrow x\times I}\right)$) $\bar{L}\subset V$ from $L$ to $L^\prime$; $\partial\bar{L} = L\coprod L^\prime$.  Thus $L$ while not uniquely defined is well-defined \emph{up to bordism}.

From Fact \ref{fact:1}, if a closed manifold $L$ of dimension $2k$ is a boundary, then $b_k(L)$ is even (and $\dim(\ker_k) = \frac{b_k(L)}{2}$).  An immediate corollary is that if $L^{2k}$ is a closed manifold which bounds, then $\chi(L) =$ even, and in particular, if $L^{2k}$ is bordant to ${L^\prime}^{2k}$, then $\chi(L)\equiv\chi\left(L^\prime\right)\mod 2$.
\end{proof}

Now $\psi_x(E)$ can be normalized to $1$ unless there is some odd length sequence of Morse transitions in $X$ starting with $E$ and returning to $E$.  That is, unless there is a submanifold $\overline{W}\subset X\times S^1$ with $\overline{W}\cap (X\times 0) = E$ and $\chi(\overline{W}) =$ odd.

But the K\"{u}nneth theorem tells us there are exactly two possible homology classes (and therefore bordism classes) for $\overline{W}$.  They are $x\times\left[S^1\right] \in H_{d-1}(X)\otimes H_1\left(S^1\right)\subset H_d\left(X \times S^1\right)$, $x\times\left[S^1\right] + [X]\times\ast \in H_{d-1}(X)\otimes H_1\left(S^1\right)\oplus H_d(X)\otimes H_1\left(S^1\right) = H_d\left(X\times S^1\right)$.  Again, to determine the parity of $\chi(\overline{W})$, it is sufficient to study any representative of each homology class, since all representatives are bordant.  In the first case, $E\times S^1$ is a representative with $\chi\left(E\times S^1\right) = 0$, so we may dismiss the first case entirely.

The second case is more interesting and fully analogous to the discussion at the close of Section \ref{sec:3}.  The class $x\times\left[S^1\right] + [X]\times\ast$ is represented by the union $E\times S^1\cup X\times 0$.  This union is not a submanifold but instead is singular along $E\times 0 = E\times S^1\cap X\times 0$.  If the normal structure, a ``
$\begin{xy}<5mm,0mm>:
(-0.5,0);(0.5,0)**\dir{-}
,(0,-0.5);(0,0.5)**\dir{-}
\end{xy}$
'' bundle over $E\times 0$, is trivial (i.e., a product bundle), the singularity can be resolved (replace
$\begin{xy}<5mm,0mm>:
(-0.5,0);(0.5,0)**\dir{-}
,(0,-0.5);(0,0.5)**\dir{-}
\end{xy}$
with
$\begin{xy}<5mm,0mm>:
(-0.5,0);(0,0.5)**\crv{(0,0)},(0,-0.5);(0.5,0)**\crv{(0,0)}
\end{xy}$) with $\chi(\text{resolution}) = \chi\left(E\times S^1\right) + \chi(X\times 0) = \chi(X\times 0)$.  While the $S^1$-direction (``vertical'') segment of ``
$\begin{xy}<5mm,0mm>:
(-0.5,0);(0.5,0)**\dir{-}
,(0,-0.5);(0,0.5)**\dir{-}
\end{xy}$
'' is trivial, the horizontal segment is reflected along $L$; thus the resolution is as previously pictured in Figure \ref{fig:2} but now in a parameter family over $L$.  Where previously a single resolution as in Figure \ref{fig:2} accounted for an additional $1\mod 2$, now, in its parametric form, it contributes $\chi(L)$.  Thus $\chi(\overline{W})\equiv \chi(X)+\epsilon(x)$.  The sector $x$ disappears from $\GDS(X)$ precisely when $\chi(\overline{W})$ is odd.

Thus completes the proof of \ref{thm:5.3} except for the statement about $Z_2$-homology spheres which is discussed in Remark \ref{rem:5.5}.
\end{proof}

To make further progress, we need the notion of ``Kervaire semi-characteristic.''

We assume $d = 2k+2$ is even.  Kervaire defined the semi-characteristic of a closed odd dimensional manifold $M^{2k+1}$
\be
s(M) = \sum_{b=0}^k b_i(M),
\ee
where $b_i$ is the $i^\text{th}$ Betti number of $M$.  Kervaire considered real Betti numbers but it is more consistent with our approach to take $Z_2$-Betti numbers:
\be
    b_i = \dim\left(H_i\left(M;Z_2\right)\right)
\ee
and consider $s(M)\in Z_2$, the mod $2$ reduction only.

The next lemma says that $s$ fluctuates in a manner similar to the sign $\theta_c$ of $H_c$.

\begin{lemma}\label{lem:5.2}
Suppose $M^{2k+1}$ undergoes a Morse transition to $N$ of index $< k$.  Then $s(N) = s(M) + 1$.  If index $=k$, this formula still holds if either
\begin{enumerate}
    \item the vanishing ($k-1$)-cycle $\delta$ at the transition is nontrivial in $H_{k-1}\left(M;Z_2\right)$ or
    \item\label{lem:5.2.2} $[\delta] = 0$ but the bordism $W$ associated to the transition has vanishing intersection pairing $H_{k+1}\left(W;Z_2\right)\times H_{k+1}\left(W;Z_2\right) \rightarrow Z_2$.
\end{enumerate}
\end{lemma}

The intersection pairing clearly vanishes on $H_{k+1}\left(W;Z_2\right)\times \image\left(H_{k+1}\left(M;Z_2\right)\right)$ so it is sufficient to consider the intersection pairing on the quotient
\[\mathcal{W} := H_{k+1}\left(W;Z_2\right) / \image H_{k+1}\left(M;Z_2\right) \cong Z_2\]
\be\label{eqn:5.3}
\mathcal{W}\times \mathcal{W}\rightarrow Z_2.
\ee

Since we are working with $Z_2$ coefficients, the vanishing of the pairing  (\ref{eqn:5.3}) is a nontrivial assumption for $W$ of all even dimensions.

\begin{proof}[Proof of \ref{lem:5.2}]
From here forward, $Z_2$ coefficients are implicit, although most assertions and diagrams have analogs over other fields and sometimes rings.

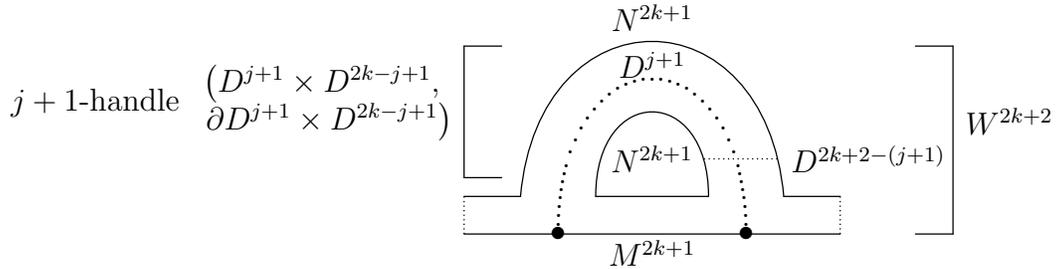
\begin{figure}[hbpt]
\[\begin{xy}<5mm,0mm>:
% center shape
(0,0)*{N^{2k+1}}
,(-1.5,-1);(1.5,-1)**\dir{-}
,(-1.5,-1);(1.5,-1)**\crv{(-1.5,2)&(1.5,2)}
% bottom line
,(-5,-2);(5,-2)**\dir{-}
% middle curve
,(-2.5,-2);(2.5,-2)**\crv{~*=<4pt>{.} (-2.5,3.5)&(2.5,3.5)}
,(-2.5,-2)*{\bullet}
,(2.5,-2)*{\bullet}
% outer curve
,(-3.5,-1);(3.5,-1)**\crv{(-3,4.5)&(3,4.5)}
% left foot
,(-3.5,-1);(-5,-1)**\dir{-}
,(-5,-1);(-5,-2)**\dir{.}
% right foot
,(3.5,-1);(5,-1)**\dir{-}
,(5,-1);(5,-2)**\dir{.}
% cut
,(1.375,0);(3.33,0)**\dir{.}
% left bracket
,(-4,3);(-5,3)**\dir{-}
,(-5,3);(-5,-0.5)**\dir{-}
,(-5,-0.5);(-4,-0.5)**\dir{-}
% right bracket
,(7,3);(8,3)**\dir{-}
,(8,3);(8,-2)**\dir{-}
,(8,-2);(7,-2)**\dir{-}
% labels
,(0,-2.5)*{M^{2k+1}}
%,(-2.5,-2.5)*{S^j}
%,(2.5,-2.5)*{S^j}
,(0,3.75)*{N^{2k+1}}
,(0,2.5)*{D^{j+1}}
,(5.75,0)*{D^{2k+2-(j+1)}}
,(-11,1.5)*{j+1\text{-handle }\begin{array}{ll}\left(D^{j+1}\times D^{2k-j+1},\right.\\
\left.\partial D^{j+1}\times D^{2k-j+1}\right)\end{array}}
,(9.5,1)*{W^{2k+2}}
\end{xy}\]
\caption{Curved dotted line $D^{j+1}$ is called ``co-core" and horizontal dotted line $D^{2k+2-(j+1)}$ is called ``core".  Solid dots indicate $\partial D^{j+1}=S^j$.}
\label{fig:5.1}
\end{figure}

The Morse transition of index $j+1$ creates a cobordism $W$ pictured in Figure \ref{fig:5.1}.  The associate ``braided'' exact sequences of pairs is displayed below.

\[\begin{xy}<5mm,0mm>:
(-11,2.25)*{0}
,(-11,1.25)*{\rotatebox{-90}{$\cong$}}
,(11,2.25)*{0}
,(11,1.25)*{\rotatebox{-90}{$\cong$}}
,(0,-5.25)*{\rotatebox{-90}{$\cong$}}
,(0,-6.25)*{Z_2}
,(-11,0)*+{H_{j+2}(W,M)}="a1"
,(-5.5,0)*+{H_{j+1}(M)}="a2"
,(0,0)*+{H_{j+1}(W,N)}="a3"
,(5.5,0)*+{H_j(N)}="a4"
,(11,0)*+{H_j(W,N)}="a5"
,(-11,-4)*+{H_{j+2}(W,N)}="b1"
,(-5.5,-4)*+{H_{j+1}(N)}="b2"
,(0,-4)*+{H_{j+1}(W,M)}="b3"
,(5.5,-4)*+{H_j(M)}="b4"
,(11,-4)*+{H_j(W,N)}="b5"
,(-2.75,-2)*+{H_{j+1}(W)}="c1"
,(8.25,-2)*+{H_j(W)}="c2"
\ar "a1";"a2"
\ar "a2";"c1"
\ar "c1";"a3"
\ar "a3";"a4"
\ar "a4";"c2"
\ar "c2";"a5"
\ar^-{\partial_{j+2}} "b1";"b2"
\ar "b2";"c1"
\ar "c1";"b3"
\ar^-{\partial_{j+1}} "b3";"b4"
\ar "b4";"c2"
\ar "c2";"b5"
\ar^-{\partial_j} "b5";"b5"+(3.5,0)
\end{xy}\]
\[H_\ast(W,M)\cong \begin{cases}Z_2 & \ast = j+1 \\ 0 & \ast\neq j+1\end{cases}\text{ and }H_\ast(W,N)\cong \begin{cases}Z_2 & \ast = 2k-j+1 \\ 0 & \ast\neq 2k-j+1\end{cases}\]

First assume $j < k$ making $H_\ast(W,N)\cong 0$, $\ast\leq j+2$.  Then $H_j(N)\cong H_j(W)$.  Suppose $\partial_{j+1} = 0$, then also $H_j(M)\cong H_j(W)$, so $H_j(N)\cong H_j(M)$ and $H_{j+1}(N)\cong H_{j+1}(W)\cong H_{j+1}(M)\oplus Z_2$.  On the other hand, if $\partial_{j+1}\neq 0$ then $H_j(M)\cong H_j(N)\oplus Z_2$ and $H_{j+1}(M)\cong H_{j+1}(N)$.  Clearly for $i\leq k$ and $i\leq j$ or $j+1$, $H_i(M)\cong H_i(W)\cong H_i(N)$.  Thus in either case, $\partial_{j+1} = 0$ or $\partial_{j+1}\neq 0$, the semi-characteristic $s\in Z_2$ satisfies $s(N) = s(M)+1$.

Now consider the more delicate case $j=k$.

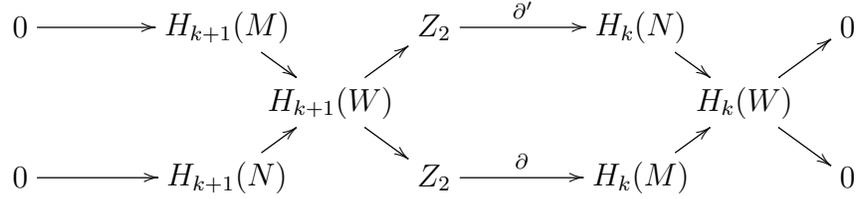
\[\begin{xy}<5mm,0mm>:
,(-11,0)*+{0}="a1"
,(-5.5,0)*+{H_{k+1}(M)}="a2"
,(0,0)*+{Z_2}="a3"
,(5.5,0)*+{H_k(N)}="a4"
,(11,0)*+{0}="a5"
,(-2.75,-2)*+{H_{k+1}(W)}="b1"
,(8.25,-2)*+{H_k(W)}="b2"
,(-11,-4)*+{0}="c1"
,(-5.5,-4)*+{H_{k+1}(N)}="c2"
,(0,-4)*+{Z_2}="c3"
,(5.5,-4)*+{H_k(M)}="c4"
,(11,-4)*+{0}="c5"
\ar "a1";"a2"
\ar "a2";"b1"
\ar "b1";"a3"
\ar^-{\partial^\prime} "a3";"a4"
\ar "a4";"b2"
\ar "b2";"a5"
\ar "c1";"c2"
\ar "c2";"b1"
\ar "b1";"c3"
\ar^-\partial "c3";"c4"
\ar "c4";"b2"
\ar "b2";"c5"
\end{xy}\]

Again consider cases.  For $\partial$, $\partial^\prime$ zero or nonzero, we have the following table for $s(N) - s(M)$:

\begin{center}\begin{tabular}{c|c|c|}
\multicolumn{1}{c}{} & \multicolumn{1}{c}{$\partial = 0$} & \multicolumn{1}{c}{$\partial\neq 0$} \\ \cline{2-3}
$\partial^\prime = 0$ & $0$ & $1$ \\ \cline{2-3}
$\partial^\prime\neq 0$ & $1$ & $\ast$ \\ \cline{2-3}
\end{tabular}\end{center}

The $\ast$ indicates that this case does not actually occur.  It is forbidden by Poincar\'{e} duality.  If $\partial(1) = y\neq 0\in H_k(M)$, $y$ will be Poincar\'{e} dual to some $\hat{y}\in H_{k+1}(M)$; $y\cdot\hat{y}=1$.  But a punctured copy of $\hat{y}$ is a null homology for $\partial^\prime(1)$ as sketched in Figure \ref{fig:5.2}.

\begin{figure}[hbpt]
\[\begin{xy}<5mm,0mm>:
% crosshairs
(-3,0)*{\bullet}
,(3,0)*{\bullet}
,(-3,0);(3,0)**\dir{-}
,(0,0.5)*{\bullet}
,(0,-0.5)*{\bullet}
,(0,0.5);(0,-0.5)**\dir{-}
% inner circles
,(4.414,0)*\xycircle<20pt>{}
,(-4.414,0)*\xycircle<20pt>{}
% labels
,(0,1.25)*{\partial^\prime(1)}
,(0,-1.25)*{\partial^\prime(1)}
,(-3.5,0)*{y}
,(3.5,0)*{y}
% left loop
,(0,0.5);(0,-0.5)**\crv{(-0.5,0.5)&(-2.75,0.25)&(-2.85,2)&(-6.5,2)&(-6.5,-2)&(-2.85,-2)&(-2.75,-0.25)&(-0.5,-0.5)}
% right loop
,(0,0.5);(0,-0.5)**\crv{(0.5,0.5)&(2.75,0.25)&(2.85,2)&(6.5,2)&(6.5,-2)&(2.85,-2)&(2.75,-0.25)&(0.5,-0.5)}
,(0,0.525);(0,-0.525)**\crv{(0.5,0.525)&(2.725,0.275)&(2.85,2.025)&(6.525,2.025)&(6.525,-2.025)&(2.85,-2.025)&(2.725,-0.275)&(0.5,-0.525)}
,(0,0.55);(0,-0.55)**\crv{(0.5,0.55)&(2.7,0.3)&(2.85,2.05)&(6.55,2.05)&(6.55,-2.05)&(2.85,-2.05)&(2.7,-0.3)&(0.5,-0.55)}
,(0,0.45);(0,-0.45)**\crv{(0.5,0.45)&(2.8,0.2)&(2.85,1.95)&(6.45,1.95)&(6.45,-1.95)&(2.85,-1.95)&(2.8,-0.2)&(0.5,-0.45)}
,(0,0.475);(0,-0.475)**\crv{(0.5,0.475)&(2.775,0.225)&(2.85,1.975)&(6.475,1.975)&(6.475,-1.975)&(2.85,-1.975)&(2.775,-0.225)&(0.5,-0.475)}
% labels
,(7.75,1.5)*{\hat{y}_\text{punctured}}
,(0,3)*{\text{handle}}
\end{xy}\]
\caption{}
\label{fig:5.2}
\end{figure}

The $(0,0)$ case occurs precisely in the forbidden case where the cobordism $W$ has a ($k+1$)-cycle with odd self-intersection, or from line (\ref{eqn:5.3}) $\mathcal{W}\times \mathcal{W}\rightarrow Z_2$ is nonzero.  One way to compute this is to take the obvious dual disks (called ``core'' and ``co-core'') $D^{k+1}\times 0$ and $0\times D^{k+1}$ in the $k+1$-handle $\left(D^{k+1}\times D^{k+1},\partial D^{k+1}\times D^{k+1}\right)$ of $W$ relative $M$ and close them off by chains, one in $M$ exhibiting a null homology for $\partial(1)$ and one in $N$ exhibiting a null homology for $\partial^\prime(1)$.  This is illustrated in Figure \ref{fig:5.3} where $W$ is drawn as a punctured M\"{o}bius band.

\begin{figure}[hbpt]
\[\begin{xy}<5mm,0mm>:
% circles
(0,0)*\xycircle<15pt>{}
,(0,0)*\ellipse<25pt>:a(150),=:a(245){}
,(0,0.5)*{M}
% normal thickness loop segments
,(-1.52,0.88);(0.75,2)**\crv{(-1.75,3)&(0.8125,2.5)}
,(1.75,1.75);(1.425,1.03)**\crv{(1.66,1.35)}
% bold loop segment 1.1
,(0.75,2);(-0.5,1.5)**\crv{(0.8125,1.375)}
,(0.725,2);(-0.5,1.525)**\crv{(0.80625,1.4)}
,(0.7,2);(-0.5,1.55)**\crv{(0.8,1.425)}
,(0.775,2);(-0.5,1.475)**\crv{(0.81875,1.35)}
,(0.8,2);(-0.5,1.45)**\crv{(0.825,1.325)}
% bold loop segment 2.1
,(1.25,2);(0.75,1)**\crv{(1.25,1.5)}
,(1.275,2);(0.775,0.9875)**\crv{(1.275,1.5)}
,(1.3,2);(0.8,0.975)**\crv{(1.3,1.5)}
,(1.225,2);(0.725,1.0125)**\crv{(1.225,1.5)}
,(1.2,2);(0.7,1.025)**\crv{(1.2,1.5)}
% bold loop segment 2.2
,(0.75,1);(-0.875,0.875)**\crv{(-0.0625,1.5)}
,(0.75,1.025);(-0.875,0.9)**\crv{(-0.0625,1.525)}
,(0.75,1.05);(-0.875,0.925)**\crv{(-0.0625,1.55)}
,(0.75,0.975);(-0.875,0.85)**\crv{(-0.0625,1.475)}
,(0.75,0.95);(-0.875,0.825)**\crv{(-0.0625,1.45)}
% labels
,(1.25,2)*{\bullet}
,(2.25,2.25)*{N}
,(-2,3);(-3,3)**\dir{-}
,(-3,3);(-3,-2)**\dir{-}
,(-3,-2);(-2,-2)**\dir{-}
,(-3.75,0.5)*{W}
,(9.25,-0.25)*{\begin{array}{l}\text{single transverse self intersection point} \\ \text{of the nontrivial element $q\in\mathcal{W}$.} \\ \text{Both dark loops represent the class $q$.}\end{array}}
% bold loop segment 1.2
\ar|(0.125)\hole@{-}@`{(0,3),(1,3),(2,2.75)} (-0.5,1.5);(1.75,1.75)
\ar|(0.125)\hole@{-}@`{(-0.025,3.025),(1,3.025),(2.025,2.775)} (-0.525,1.5);(1.775,1.75)
\ar|(0.125)\hole@{-}@`{(-0.05,3.05),(1,3.05),(2.05,2.8)} (-0.55,1.5);(1.8,1.75)
\ar|(0.125)\hole@{-}@`{(0.025,2.975),(1,2.975),(1.975,2.725)} (-0.475,1.5);(1.725,1.75)
\ar|(0.125)\hole@{-}@`{(0.05,2.95),(1,2.95),(1.95,2.7)} (-0.45,1.5);(1.7,1.75)
% bold loop segment 1.3
\ar@{-}@/_2pt/ (1.75,1.75);(0.75,2)
\ar@{-}@/_2pt/ (1.75,1.775);(0.75,2.025)
\ar@{-}@/_2pt/ (1.75,1.8);(0.75,2.05)
\ar@{-}@/_2pt/ (1.75,1.725);(0.75,1.975)
\ar@{-}@/_2pt/ (1.75,1.7);(0.75,1.95)
% bold loop segment 2.3
\ar|(0.45)\hole@{-}@`{(-1.125,2.625),(1.375,3.125)} (-0.875,0.875);(1.25,2)
\ar|(0.45)\hole@{-}@`{(-1.1125,2.675),(1.3875,3.15)} (-0.9,0.875);(1.275,2)
\ar|(0.45)\hole@{-}@`{(-1.1,2.725),(1.4,3.175)} (-0.925,0.875);(1.3,2)
\ar|(0.45)\hole@{-}@`{(-1.2,2.575),(1.3625,3.1)} (-0.85,0.875);(1.225,2)
\ar|(0.45)\hole@{-}@`{(-1.15,2.55),(1.35,3.075)} (-0.825,0.875);(1.2,2)
% arrow to bullet
\ar (2,1);(1.45,1.8)
\end{xy}\]
\caption{}
\label{fig:5.3}
\end{figure}

\end{proof}

\begin{remark}\label{rem:5.5}
This remark completes the proof of Theorem \ref{thm:5.3}; the statement about manifolds $X$ with $H_\ast(X)\cong H_\ast\left(S^d\right)$.  When $H_{d-1}(X)\cong 0$, each elementary bordism $W$ across a plaquette (i.e., $d$-cell) will necessarily be orientable and hence (unlike the bordism in $T^2$ between $(1,1)$ and $(1,-1)$) be represented as a codimension zero submanifold of $X$, $W\subset X$.  This and $H_{\frac{d}{2}}(X)\cong 0$ implies that even when Morse transitions are of index $k$, $d = 2k+2$, condition (\ref{lem:5.2.2}) of (\ref{lem:5.2}) will hold, whether or not $[\delta] = 0$.  So we conclude that \emph{all} Morse transitions of $E\subset X$ change the parity of $s(E)$.  Furthermore, in the previous part of the proof we established that $\chi(\overline{W}) =$ even in the present case since $\chi(X) = \chi\left(S^d\right) = 2$ and $\epsilon(0) = 0$.  These two facts imply that $(-1)^{s(E)}$ is a consistent assignment of signs for a nontrivial (and unique up to scale) zgswf in $\GDS(X)$. This completes the last portion of Theorem \ref{thm:5.3}.  Note that the proof only requires $H_{d-1}(X) = 0 = H_\frac{d}{2}(X)$, a weaker assumption than assuming $X$ is a $Z_2$-homology sphere.
\end{remark}

\begin{remark}\label{qcremark}
For $d$ odd, Theorem \ref{thm:5.3} gives an explicit form of the ground state wavefunction in each sector as
$\psi_x(E) = i^{\chi(E)}$.
We claim that the unitary $U$ that multiplies each basis vector by a phase $i^{\chi(E)}$ can be written as a local quantum circuit, i.e.,
with a bounded depth of the circuit and with the gates of the circuit supported on sets of bounded diameter, with these bounds depending only on local geometry.
Thus, there is a local quantum circuit that maps the ground state subspace of the $\GDS$ for $d$ odd onto the ground state subspace of a $\GTC$.
Further, since every Morse transition for odd $d$ changes $\chi$ by $\pm 2$, if we restrict to the ground state subspace of $H_+$ then the unitary $U$ conjugates the operator $H_\square$ of the $\GDS$ to the operator $H_\square$ of the $\GTC$ (we leave it as an open problem whether there is a relationship between the theories if we consider the space of states outside the ground state subspace of $H_+$).

To construct the local quantum circuit, write
\be
U=\prod_{j=0}^{d-1} \prod_\text{$j$-cells c} U(j,c),
\ee
where the second product is over all $j$-cells $c$, and $U(j,c)$ is a unitary that
 that multiplies basis vectors by phase equal to either $1$ or $\pm i$, with the phase being $1$ if the given cell $c$ is not in $E$ and $\pm i$ if $c$ is in $E$ (the phase being $+i$ if $j$ is even and $-i$ if $j$ is odd).  Each $U(j,c)$  is supported only on the $(d-1)$ cells that intersect the given $j$-cell $c$.  Those unitaries are the gates in the quantum circuit.  Unitaries supported on disjoint sets can be performed in the same round, and in a bounded number of
rounds all unitaries can be performed, assuming the cellulation has bounded local geometry, i.e., assuming there is a bound on the number of $(j+1)$ cells meeting any $j$ cell and also a bound on the number of $j$-cells in the boundary of any $(j+1)$ cell and a bound on the diameter of any cell.
\end{remark}

\section{Generalized loop, or ``balloon'' operators}\label{sec:6}

Let us recall the Wilson loop ``$W_l$'' operators \cite{LW} for the usual $\TC$ and $\DS$ theories.  In the $\TC$ case, one merely specifies a loop $l$, the support of $W_l$, and then $W_l$ acts as $\operatorname{NOT}$ on the bonds of $l$.  In the case of $\DS$, this definition of $W_l$ is augmented by a sign rule (e.g., line 44 of \cite{LW}) which is necessary to preserve ground states.  The rule stated in \cite{LW} for a loop $l$ or the $2$-sphere $S^2$ may be restated in topological language as the sign rule below.

Let $\alpha$ be the multi-loop of $\uparrow$ bonds.  The sign for the action of $W_l$ on $\psi(\alpha)$ is:

\be
-1^{s(l)} i^{\chi(\partial(l\cap\alpha))} \left(-1^{\operatorname{link}(\partial(l\cap\alpha))}\right)
\ee

Again $s$ is $Z_2$-semi-characteristic, which is $1$ for a simple closed curve $l$.  $(l\cap\alpha)$ consists of line segments (or all of $l$ if $l\subset\alpha$), so $\partial(l\cap\alpha)$ is an even number of points and $\chi$ merely counts them.  These points are naturally divided into pairs (``$0$-spheres'') as the boundaries of the component arcs of $\alpha\setminus l$. For a collection of $0$-spheres in a circle $l$, there is a total mod $2$ linking number which we denote by ``link,'' which can be computed, for example, by making $l$ bound a disk $\Delta$, then spanning all the $0$-spheres by simple arcs in $\Delta$, and finally counting (mod $2$) the number of pairwise intersections of the arcs.  An equivalent
definition is to reflect all component arcs in $\alpha \setminus l$ which lie outside $\Delta$ into the interior of $\Delta$, and then count pairwise intersections.
An example is given in Figure \ref{fig:6.1}.

\begin{figure}[hbpt]
\[\begin{xy}<5mm,0mm>:
% circle
(0,0)*\xycircle<20pt>{}
,(-3.25,0)*{l\text{ (round)}}
% right arc
,(0,0)*\ellipse<20.35pt>:a(-135),=:a(180){}
,(0,0)*\ellipse<20.7pt>:a(-135),=:a(180){}
,(0,0)*\ellipse<19.65pt>:a(-135),=:a(180){}
,(0,0)*\ellipse<19.3pt>:a(-135),=:a(180){}
% left arc (0,1.414);(-1,1)**\crv{(-0.5,1.414)}
,(0,1.459);(-1,1.05)**\crv{(-0.5,1.459)}
,(0,1.43);(-1,1.025)**\crv{(-0.5,1.43)}
,(0,1.364);(-1,0.95)**\crv{(-0.5,1.364)}
,(0,1.389);(-1,0.975)**\crv{(-0.5,1.389)}
% middle arc
,(0.25,0)*{\alpha}
,(0,1.414);(-1,-1)**\crv{(0,0)}
,(0.05,1.414);(-0.95,-1.025)**\crv{(0.05,0)}
,(0.025,1.414);(-0.975,-1.0125)**\crv{(0.025,0)}
,(-0.05,1.414);(-1.05,-0.975)**\crv{(-0.05,0)}
,(-0.025,1.414);(-1.025,-0.9875)**\crv{(-0.025,0)}
% hat
,(-1,1);(1,1)**\crv{(-0.375,2.25)&(0.5,2.375)&(1.25,2)}
,(-1.025,1);(1.025,1)**\crv{(-0.4,2.275)&(0.5,2.4)&(1.275,2.025)}
,(-1.05,1);(1.05,1)**\crv{(-0.4,2.3)&(0.5,2.425)&(1.3,2.05)}
,(-0.95,1);(0.95,1)**\crv{(-0.4,2.2)&(0.5,2.325)&(1.2,1.95)}
,(-0.975,1);(0.975,1)**\crv{(-0.4,2.225)&(0.5,2.35)&(1.225,1.975)}
,(3.25,2)*{\alpha\text{ (in bold)}}
\end{xy}\hspace{.1in}
\begin{xy}
\ar^{W_l} (0,0);(15,0)
\end{xy}\hspace{.1in}
\begin{xy}<5mm,0mm>:
% circle
(0,0)*\xycircle<20pt>{}
% left arc
,(-1,-1);(-1,1)**\crv{(-1.75,0)}
,(-0.98,-0.98);(-0.98,0.98)**\crv{(-1.725,0)}
,(-0.96,-0.96);(-0.96,0.96)**\crv{(-1.7,0)}
,(-1.04,-1.04);(-1.04,1.04)**\crv{(-1.8,0)}
,(-1.02,-1.02);(-1.02,1.02)**\crv{(-1.775,0)}
% middle arc
,(0.25,0)*{\alpha}
,(0,1.414);(-1,-1)**\crv{(0,0)}
,(0.05,1.414);(-0.95,-1.025)**\crv{(0.05,0)}
,(0.025,1.414);(-0.975,-1.0125)**\crv{(0.025,0)}
,(-0.05,1.414);(-1.05,-0.975)**\crv{(-0.05,0)}
,(-0.025,1.414);(-1.025,-0.9875)**\crv{(-0.025,0)}
% right arc ((0,1.414);(1,1)**\crv{(0.5,1.414)})
,(0,1.459);(1,1.05)**\crv{(0.5,1.459)}
,(0,1.43);(1,1.025)**\crv{(0.5,1.43)}
,(0,1.364);(1,0.95)**\crv{(0.5,1.364)}
,(0,1.389);(1,0.975)**\crv{(0.5,1.389)}
% hat
,(-1,1);(1,1)**\crv{(-0.375,2.25)&(0.5,2.375)&(1.25,2)}
,(-1.025,1);(1.025,1)**\crv{(-0.4,2.275)&(0.5,2.4)&(1.275,2.025)}
,(-1.05,1);(1.05,1)**\crv{(-0.4,2.3)&(0.5,2.425)&(1.3,2.05)}
,(-0.95,1);(0.95,1)**\crv{(-0.4,2.2)&(0.5,2.325)&(1.2,1.95)}
,(-0.975,1);(0.975,1)**\crv{(-0.4,2.225)&(0.5,2.35)&(1.225,1.975)}
,(1.5,2)*{\alpha}
\end{xy}\]
\[\text{sign} = (-1)\left(i^4\right)\left(-1^1\right) = 1\]
\caption{}
\label{fig:6.1}
\end{figure}

Our task in this section is to define the correct signs for the analogous Wilson ``balloon'' operator.  To obtain explicit formulas, we introduce restrictions on the topology of $X$ as necessary.  These operators are equal to $\pm \operatorname{NOT}$ on an imbedded codimension $=1$ $\mathcal{C}$-cellular submanifold $L^{d-1}\subset X^d$, up to a sign.  To be consistent, the sign $\pm$ must be set so as to carry zgswfs to zgswfs.  Keeping similar notation, $\alpha$ is the ($d-1$)-dimensional cellular submanifold supporting a basis vector and $L$ the Wilson balloon.

Note that in addition to these Wilson loop operators $W_l$, the operator given by the product of Pauli $Z$ around any closed loop on the dual lattice also commutes with the Hamiltonian of both the TC and DS models.  This operator (which we will call a dual Wilson loop) straightfowardly generalizes to the $\GTC$ and $\GDS$.  For the $\GDS$, given any closed loop on the dual lattice, the product of Pauli $Z$ over all $(d-1)$-cells that intersect that loop is an operator that commutes with the Hamiltonian.  If the closed loop is nullhomologous, then this operator acts as the identity on the ground state subspace.

\begin{theorem}\label{thm:6.1}
Let $d=2k+2$, $k>0$.  As in \ref{thm:5.3}, we require that $X^d$ is a $Z_2$- homology sphere.\footnote{It is sufficient to assume $H_{d-1}(X) \cong 0\cong H_\frac{d}{2}(X)$ and that the Stiefel-Whitney classes $w_i$ of the tangent bundle $\tau(X)$ vanish for $i\leq\frac{d}{2}$.}  Without this assumption, a local description of the balloon operator signs may not be possible.  The consistent sign for Wilson balloon operators is:
\[(-1)^{s(L)}\left(i^{\chi(\partial(L\cap\alpha))}\right)\]
For the case of odd dimensional $d$, let $d=2k+1$.  Then for all closed $X^d$ (no restriction on topology), the consistent sign is:
\[i^{\chi(L)} (-1)^{\chi(L\cap\alpha)}\]
\end{theorem}

\subsubsection*{Note}
The formula in odd dimensions is quite elementary, involving only Euler characteristic.  In even dimensions, the formula is curiously simpler than the classical $d=2$ case of $\DS$.  If one traces the reason back through the proof, the anomaly for $d=2$ is that $\partial(L\cap\alpha)$ is then zero dimensional, i.e. points. In this very low dimension, the middle dimensional (zero) homology of $S^0\times S^0$ a product of spheres has rank $4$ not $2$ and the usual analysis of ``hyperbolic pairs,'' forms $\left|\begin{array}{cc} 0 & 1 \\ \pm 1 & 0\end{array}\right|$ or over $Z_2$ simply $\left|\begin{array}{cc} 0 & 1 \\ 1 & 0\end{array}\right|$, does not apply.

\begin{proof}[Proof of \ref{thm:6.1}]
Consider first $d$ odd.  Write $L=A\cup B$, and $\alpha = B\cup C$, $M = \partial A = \partial B = \partial C$, as shown in Figure \ref{fig:6.2}, and apply the additivity formula (\ref{eqn:5.4}) for Euler characteristic.

\begin{figure}[hbpt]
\[\begin{xy}<5mm,0mm>:
(0,0)*\xycircle<20pt>{}
,(0,1.4);(1.4,0)**\crv{(1.4,3)&(3,2)}
,(1.4,0)*{\bullet}
,(2,-0.125)*{M}
,(0,1.4)*{\bullet}
,(-0.125,2)*{M}
,(-0.66,-0.66)*{A}
,(0.66,0.66)*{B}
,(2.25,2.25)*{C}
\end{xy}\]
\caption{}
\label{fig:6.2}
\end{figure}

\be
\begin{array}{cc}
\chi(A\cup C) - \chi(B\cup C) = \chi(A) + \chi(C) - \chi(M) - \chi(B) - \chi(C) + \chi(M) = \\
\chi(A) - \chi(B)= \chi(A) + \chi(B)-2\chi(B) = \chi(A \cup B) + \chi(M)-2\chi(B) = \\
\chi(L)+\chi(M)-2\chi(L \cap \alpha).
\end{array}
\ee
Since $M$ is a closed odd dimensional manifold, $\chi(M)=0$.
Thus the change in sign of a configuration in the gswf upon applying $W_L$ is as claimed.

Now consider $d=2k+2$, even.  To handle this case, we need a rather powerful algebraic lemma \ref{lem:6.3} which, since we are using $Z_2$-coefficients, may be stated uniformly regardless of whether $k$ is even or odd.  Lemma \ref{lem:6.3} will be applied to the three kernels
\[\ker\left(H_k(M)\rightarrow H_k(A)\right)\text{, }\ker\left(H_k(M)\rightarrow H_k(B)\right)\text{, and }\ker\left(H_k(M)\rightarrow H_k(C)\right)\]
which, in a slight abuse of notation, are also denoted $A$, $B$, and $C$, respectively.

\begin{lemma}\label{lem:6.2}
If the Stiefel-Whitney classes of the tangent bundle $w_i\left(\tau\left(X^{2k+2}\right)\right)$ vanish for $i\leq k = \left(\frac{d}{2} - 1\right)$, then the intersection form on $H_k\left(M^{2k}; Z\right)$ is isomorphic to a direct sum $Q = \bigoplus_{i=1}^j \left(\begin{array}{cc} 0 & 1 \\ 1 & 0\end{array}\right)$.
\end{lemma}

\begin{proof}
The total Stiefel-Whitney class of $\tau(M)$, $w = 1 + w_1 + w_2 +\cdots + w_{2k}$, satisfies $w = \operatorname{Sq}(v)$ where $v = 1 + v_1 + v_2 +\cdots + v_{2k}$ is the total Wu class and $\operatorname{Sq}$ is the total Steenrod square $\operatorname{Sq} = 1 + \operatorname{Sq}^1 + \operatorname{Sq}^2 + \cdots$.  Inductively we find:
\[w_1 = v_1\]
\[w_2 = v_2 + v_1\cup v_1\]
\[w_3 = v_3 + {\rm Bockstein}(v_2)\]
\[w_4 = v_4 +{\rm Bockstein}(v_3) + v_2\cup v_2\]
\[\vdots\]
so $w_i = 0$, $i\leq k$, implies $v_i = 0$, $i\leq k$.

We have assumed that $w_i(\tau(X)) = 0$, $i\leq k$.  However, $\tau(X)\vert_M = \tau(M)\oplus \nu_{M\hookrightarrow X}$ and $\nu_{M\hookrightarrow X}$ is trivialized by the cross-sections (inward normal to $A$, inward normal to $B$).  Thus, for all $i$, $w_i \tau(M) = \operatorname{inc}^\ast w_i(\tau(X))$, implying that $w_i(\tau(M)) = 0$, $i\leq k$.  Hence the $k^\text{th}$ Wu class of $M$, $v_k = 0\in H^k(M)$, vanishes.

$\operatorname{Sq}^k(x) = v_k\cup x$ for $x\in H^k(M)$ so:
\[x\cup x = v_k\cup x,\]
so the vanishing of $v_k$ implies that for all $x\in H^k(M)$, $x\cup x = 0$.  It is now an elementary exercise in duality to put the intersection form of $M$ in the claimed coordinates.
\end{proof}

In Lemma \ref{lem:6.3}, $Z_2$ is called $F_2$, the field of two elements, to emphasize that it is a field.

\begin{lemma}\label{lem:6.3}
Consider a space $F_2^{2j}$, with coordinates labelled $1,2,...,2j$.  Let $Q$ be the form
\be
Q=\begin{pmatrix}
0 & 1 \\
1 & 0 \\
&& 0 & 1 \\
&&1 & 0 \\
&&&& \ldots
\end{pmatrix}.
\ee
Let $A,B,C$ be three maximal subspaces (i.e., $j$-dimensional subspaces) such that $Q$ vanishes on each subspace.
Assume also that for any triple of vectors $a\in A, b \in B, c \in C$ such that $a+b+c=0$ we have \be
\label{ztwist}
(a,Qb)=0.
\ee
Then
\be
\dim(A \cap B)+\dim(B \cap C)+\dim(C \cap A)=j \mod 2.
\ee
\end{lemma}

Since we will give an inductive proof and we will want to show that this last condition Eq.~(\ref{ztwist}) holds under the induction, we will be refering to it below; so, we give it a name, referring to it as
the ``no twist" condition.
Before giving the proof, we remark that the no twist condition is symmetric under interchange of the subspaces $A,B,C$, because $(b,Qc)=(b,Qa)+(b,Qb)=0$ and similarly $(a,Qc)=0$.

\begin{proof}
Assume that $\dim(A\cap B)\neq 0$.  In this case we will give an inductive proof, reducing the problem to a problem with $j$ reduced by $1$.
Apply a basis transformation so that $A$ and $B$ both contain the vector $v_1 \equiv (1,0,0,0,\ldots)$.  Then, every vector in $A,B$ vanishes on the second entry.
So, $A$ can be written as the span of $v_1$ and $A'$ and $B$ can be written as the span of $v_1$ and $B'$, where $A'$ and $B'$ are $j-1$-dimensional subspaces with $Q$ vanishing on $A'$ and $B'$ and every vector in $A'$ or $B'$ vanishing on the first two entries.
Suppose that $v_1$ is also in $C$ or that $v_2=(0,1,0,0,\ldots)$ is in $C$.  Then, $C$ can also be written as the span of $v_1$ and $C'$ or $v_2$ and $C'$, with $C'$ vanishing on the first two entries.  Then, the subspaces $A',B',C'$ are maximal subspaces of a space $F_2^{2(j-1)}$ with $Q$ vanishing on $A',B',C'$ and the no twist condition holding for $A',B',C'$  because $A',B',C'$ are subspaces of $A,B,C$ respectively.  Further,
\be
\label{inductive}
\begin{array}{ll}
\dim(A \cap B)+\dim(B \cap C)+\dim(C \cap A)=\\
\dim(A' \cap B')+\dim(B' \cap C')+\dim(C' \cap A')+1 \mod 2,
\end{array}
\ee
as desired.

So, suppose that $v_1$ and $v_2$ are not in $C$.  Then, it cannot be the case that every vector in $C$ vanishes on the first coordinate (else $v_2$ must be in $C$) or that every vector in $C$ vanishes on the second coordinate (else $v_1$ must be in $C$).  So, one of two cases hold.  Either $C$ contains a vector $w=(1,0,w_3,w_4,\ldots)$ and a vector $w'=(0,1,w'_3,w'_4,\ldots)$ or
$C$ contains a vector $x=(1,1,x_3,x_4,\ldots)$ but does not contain any vector of the form $w$ or $w'$.

Suppose the first case holds, that $C$ contains such a $w$ and $w'$.
So, $C$ can be written as the span of $w,w'$ and some subspace $C'_0$ that vanishes on the first and second coordinates.  Let $y$ be the vector $(0,0,w_3,w_4,\ldots)$.  Define $C'$ to be the span of $C'_0$ and $y$.  Note that $y$ is not in $C'_0$ by the assumption that $v_1$ is not in $C$ so $C'$ has dimension $j-1$.  Note that $\dim(A \cap C)=\dim(A \cap {\rm span}(w,C_0'))=\dim(A' \cap C')$ and similarly $\dim(B' \cap C')=\dim(B \cap C)$.
So, this again gives three subspaces $A',B',C'$ for which Eq.~(\ref{inductive}) holds and the no twist condition holds so again we reduce $j$ by $1$.

Suppose instead the second case holds, that $C$ contains a vector
\[x=(1,1,x_3,x_4,\ldots)\]
but does not contain any vector of the form $w$ or $w'$.
Define $C'$ to be the projection of $C$ onto the space on which the first two coordinates vanish; i.e., $C'$ is the space of vectors $(0,0,v_3,v_4,\ldots)$
such that $(0,0,v_3,v_4,\ldots)$ is in $C$ or $(1,1,v_3,v_4,\ldots)$ is in $C$.  Being a projection, $C'$ has dimension $j$ or $j-1$.  The form $Q$ vanishes on $C'$ so $C'$ has dimension $j-1$.  So, if a vector $(1,1,x_3,x_4,\ldots)$ is in $C$ then $(0,0,x_3,x_4,\ldots)$ is in $C$.  So, $(1,1,0,0,\ldots)$ is in $C$ so $C$ is the span of $(1,1,0,0,\ldots)$ and $C'$, so $\dim(A' \cap C')=\dim(A \cap C)$ and $\dim(B' \cap C')=\dim(B \cap C)$ so Eq.~(\ref{inductive}) holds.
Further, the no twist condition holds for $A',B',C'$.

So, we succeed in reducing $j$ by $1$ whenever $\dim(A\cap B)\neq 0$.  A similar reduction can be performed whenever $\dim(A \cap C) \neq 0$ or $\dim(B \cap C)\neq 0$ so
we can assume that $A\cap B=B\cap C=C\cap A=\emptyset$.  We now show that this implies that $j$ is even.

Make a basis transformation leaving $Q$ invariant so that $A$ is the space of vectors vanishing on the even coordinates.
Introduce notation, writing a vector $v\in F_2^{2j}$ as $v=(x;y)$ where $x,y$ are $j$-dimensional vectors to indicate that the odd coordinates of $v$ are the entries of $x$ while
the even coordinates are the entries of $y$.  Parametrize $B$ as the space of vectors $(Ku;Lu)$ for two matrices $K,L$ where $u$ is a $j$-dimensional vector.  Since $B$ is $j$-dimensional and
$A \cap B=\emptyset$, $L$ is non-singular.  So, we can instead parametrize $B$ as the space of vectors $(KL^{-1} y; y)$ or $(My;y)$ where $M=KL^{-1}$.  Similarly, we can parametrize $C$
as the space of vectors $(Ny;y)$.  The matrices $M,N$ are symmetric because $Q$ vanishes on $B,C$.  We now use the no twist condition.  We have $((M+N)y;0)+(My;y)+(Ny;y)=0$, and so $(My,y)+(Ny,y)=0$ for all $y$ or equivalently $(y,(M+N)y)=0$.  So, $M+N$ vanishes on the diagonal.  However, $M+N$ must be nonsingular (otherwise, $B \cap C \neq \emptyset$).
So, $j$ must be even.

The fact that $j$ must be even follows from the fact that any symmetric, non-singular matrix over $F_2$ which vanishes on the diagonal must be even dimensional.  To show this, consider any such matrix $F$.  If $G$ is a non-singular matrix, then $G^T F G$ is still a symmetric non-singular matrix which vanishes on the diagonal.  We use two different choices of $G$, either a permutation matrix or a matrix which is $1$ on the diagonal and has one nonzero entry of the diagonal, to simplify the matrix $F$ through a sequence of steps similar to Gaussian elimination; the second choice of $G$ gives a ``row and column" operation.  $F$ must have a nonzero entry on the first row (otherwise it is singular), so use permutations to make the upper right corner of $F$ nonzero.  Then, use ``row and column" operations to make all entries in the last column vanish, except for the entry in the upper right corner.  Repeat this procedure on the second row, and so on, continuing until all the nonzero entries are on the diagonal from top right to bottom left corner.  This diagonal intersects the main diagonal if $j$ is odd (and so the matrix must be singular), so $j$ must be even.
\end{proof}

Again, the notation of the lemma is designed to evoke Figure \ref{fig:6.2}, with a minimal abuse of notation.  We wrote $A$, $B$, and $C$ for the kernels of the three maps:
\[H_k(M)\rightarrow H_k(A)\]
\[H_k(M)\rightarrow H_k(B)\]
\[H_k(M)\rightarrow H_k(C).\]
As we now show, these kernels indeed have the ``no twist'' (nt) property:
\be\label{eqn:nt}
\tag{\text{nt}}
\text{$a+b+c=0$ ($a\in A$, $b\in B$, $c\in C$) implies $(a,Qb) = 0$.}
\ee

\begin{lemma}\label{lem:6.4}
Under the assumption $w_i(\tau(X)) = 0$, $i\leq k+1 = \frac{d}{2}$, (\ref{eqn:nt}) holds.
\end{lemma}

\begin{proof}
First recall from the proof of Lemma \ref{lem:6.2} that $w_i(\tau(X)) = 0$, $1\leq i\leq k+1$, implies that two Wu classes vanish, $v_{k+1}\subset H^{k+1}(X)$ and $v_k\in H^k(M)$.  The first vanishing is equivalent to $x\cdot x = 0$ for all $x\in H_{k+1}(X)$ and the second vanishing is equivalent to $y\cdot y = 0$ for all $y\in H_k(M)$.

Assuming $(a,Qb)=1$, we may build a cycle $x\in H_j(X)$ with $x\cdot x=1$, a contradiction.  The construction is to form $V$ representing $x$ by gluing together three coboundaries: $V = \delta_a^A\cup\delta_b^B\cup\delta_c^C$ where the superscript is the ambient space of the coboundary, and $\partial\left(\delta_p^- = p\right)$.

\begin{figure}[hbpt]
\[\begin{xy}<5mm,0mm>:
% left circles
(0,-0.5)*\xycircle<10pt>{}
,(-0.75,0.25)*{a^\prime}
,(0,0)*\ellipse<20pt>:a(-35),=:a(340){}
,(-1.25,1.25)*{a}
% left cones
,(-1.25,1);(-3.5,0)**\dir{-}
,(-3.5,0);(-1.75,-0.75)**\dir{-}
,(-2.5,-0.625);(-3.5,-2)**\dir{-}
,(-3.5,-2);(-0.75,-1.75)**\dir{-}
% left middle cone
,(0,0.225);(0.5,3)**\dir{-}
,(0.5,3);(3.25,-0.75)**\dir{-}
% right cones
,(5,0.25);(7,0)**\dir{-}
,(7,0);(5,-0.25)**\dir{-}
,(4.25,-0.5);(6.25,-0.75)**\dir{-}
,(6.25,-0.75);(4.25,-1)**\dir{-}
% labels
,(-3,1)*{\delta_a^A}
,(-3,-2.75)*{\delta_{a^\prime}^A}
,(8,0)*{\delta_b^B}
,(7.25,-0.75)*{\delta_{b^\prime}^B}
,(0.5,4)*{\delta_{c^\prime}^C}
,(2.5,4)*{\delta_c^C}
,(4.5,0)*{b}
,(3.75,-0.75)*{b^\prime}
,(3,-3)*{\begin{array}{c}\text{this crossing leads to an}\\\text{intersection point of }x\cdot x\end{array}}
\ar (1,-2);(1,-1.25)
% right middle cone
\ar|(0.35)\hole@{-} (0.6875,1.225);(2.5,3)
\ar@{-} (2.5,3);(4,0)
% top ellipse
\ar|(0.45)\hole@{-}@`{(1.4,0.4),(4,0.4)} (1.4,0);(4,0)
\ar|(0.55)\hole@{-}@`{(1.4,-0.4),(4,-0.4)} (1.4,0);(4,0)
% bottom ellipse
\ar|(0.325)\hole@{-}@`{(0.67,-0.35),(3.25,-0.35)} (0.67,-0.75);(3.25,-0.75)
\ar@{-}@`{(0.67,-1.15),(3.25,-1.15)} (0.67,-0.75);(3.25,-0.75)
% line between ellipses
\ar@{.} (1.4,0);(0.67,-0.75)
\end{xy}\]
\caption{}
\label{fig:6.3}
\end{figure}

\iffalse
To compute $x\cdot x$, we break symmetry and form the closed $(d-1)$-hypersurface $A\cup B$ and, using the product collars of $\delta$ in $A$ and $\delta$ in $B$, displace $\delta_a^A\cup\delta_b^B$ toward $A$.  Call the result ${\delta_a^B}^\prime\cup{\delta_b^B}^\prime$ (see Figure \ref{fig:6.3} for the displacement).  The crossing shown in Figure \ref{fig:6.3} translates to a single transverse intersection point $q$ of $\delta_a\cup\delta_b$ and $\delta_a^\prime\cup\delta_b^\prime$ once both are deformed, relative boundary, into a normal product collar $K = \mathpzc{n}\times I$, a neighborhood, $\mathpzc{n}$, of $A\cap B$ within $A\cup B$.  $\mathpzc{n}$ is taken to be oppositely directed from the inward normal over $A\cap B$ into $C$.

To finish computing $x\cdot x$, perturb a copy $V^\prime$ of $V$ within $A\cup B\cup C$, $\overline{V} := \delta_a^{A^\prime}\cup\delta_b^{B^\prime}\cup\delta_c^{C^\prime}$, before perturbation into transverse position within $X$, $\widetilde{V} \pitchfork \widetilde{V}\subset X$.  Initially, before perturbation, $V$ and $V^\prime$ intersect in an arc $\gamma$ and a closed $1$-manifold $S\subset V\subset A\cup B\cup C =: T$.  $\gamma$ is associated to the crossing point in Figure \ref{fig:6.3} and leads to the intersection point $q$ already identified.  Additional intersection points $q^\prime\subset x\cdot x$ come from perturbing $S$ normally to $T$.  Near $\delta := A\cap B = B\cap C = C\cap A$, cyclically symmetric normal directions are defined (Figure \ref{fig:6.4}).
\fi

The plan is to compute $x\cdot x$ taking $V$ and a copy $\widetilde{V}$ perturbed in $X$ to be in general position with respect to $V$ and counting intersection points.  The answer is $x\cdot x = \left|V\cap\widetilde{V}\right| = (a, Qb)$, but from our assumptions on $X$, $x\cdot x = 0$, implying $(a, Qb) = 0$.

Although the submanifolds $A$, $B$, and $C$ play symmetric roles, it is helpful for the computation to \emph{break symmetry} and consider $a\cup b\subset L := A\cup B$.  We may obtain a second copy $a^\prime\cup b^\prime$ of $a\cup b$ by perturbing this $k$-cycle toward $A$ using the normal direction of $M\subset A$.  Let $\delta_a^{A^\prime}$, $\delta_b^{B\prime}$, and $\delta_c^{C^\prime}$ be formed by extending this perturbation over the three $\delta$s.  There is a $Z_2$-linking number defined:
\[l\left(a\cup b, a^\prime\cup b^\prime\right) = (a\cup b)\cdot\left(\delta_a^{A^\prime}\cup\delta_b^{B^\prime}\right) = \left(\delta_a^A \cup\delta_b^B\right) \cdot\left(a^\prime\cup b^\prime\right).\]

As in the construction of the classical Seifert form in knot theory,
\[l\left(a\cup b, a^\prime\cup b^\prime\right) = l\left(a,a^\prime\right) + l\left(a,b^\prime\right) + l\left(b,a^\prime\right) + l\left(b,b^\prime\right) :=\]
\[a\cdot\delta_a^{A^\prime} + a\cdot\delta_b^{B^\prime} + b\cdot\delta_a^{A^\prime} + b\cdot\delta_b^{B^\prime} = 0 + [a]\cdot[b] + [b]\cdot[b] = 0 + (a,Qb) + 0\]
where the first term is zero since the perturbation is towards $A$, the $[\;\;]$ denote classes in $H_k(M)$, and the ``dot'' on the last line denotes intersection in $M$.  The last term is zero since $V_k = 0 \in H^k(M)$.

Let $V^\prime = \delta_a^{A^\prime}\cup\delta_b^{B^\prime}\cup\delta_c^{C^\prime}$.  $V\cap V^\prime$ consists of a collection of arcs counted by $a\cdot b\equiv_2 (a, Qb)$. Now form $\widetilde{V}$ by further perturbing $V^\prime$ to be transverse to $V$ in $X$.  Each arc of $V\cap V^\prime$ becomes a transverse intersection point of $V\cap \widetilde{V}$.  Thus we conclude $x\cdot x\equiv_2\left|V\cap\widetilde{V}\right|\equiv_2 (a, Qb)$, provided all \emph{other} points of $V\cap\widetilde{V}$ come in \emph{pairs}.

Before we perturbed $V^\prime$ to $\widetilde{V}$, $V\cap V^\prime$ consists of the just described arcs (counted by $a\cdot b$) and a closed $1$-manifold $S\subset V\subset A\cup B\cup C$.  We denote $A\cup B\cup C$ by $T$.  Any additional intersection points arise as the zero locus of a normal perturbation to $S$ within the ($d-1$)-manifold sheet of $T$ in which it \emph{locally} lives.

\begin{figure}[hbpt]
\[\begin{xy}<5mm,0mm>:
% sheet C
(0,0);(4,1)**\dir{-}
,(4,1);(4,3)**\dir{-}
,(4,3);(0,2)**\dir{-}
,(0,2);(0,0)**\dir{-}
,(3.5,2.375)*{C}
% sheet B
,(0,0);(2,-2)**\dir{-}
,(2,-2);(6,-1)**\dir{-}
,(6,-1);(4,1)**\dir{-}
,(1.5,-0.5)*{B}
% sheet A
,(4,1.25);(2.5,2)**\dir{--}
,(2.5,2);(0.16,1.415)**\dir{--}
,(-0.16,1.335);(-1.5,1)**\dir{-}
,(-1.5,1);(0,0)**\dir{-}
,(-0.5,0.875)*{A}
% line across A and B
,(0.5,1.5);(4,-1.5)**\crv{(1.6,0.7)&(2,0.5)&(2.4,0.3)}
% line across B and C
,(2,2.5);(4.75,-1.3125)**\crv{(2,1.5)&(2.065,0.9375)}
% delta
,(-1,-1)*{\delta}
% S arrows
,(4.625,-2.5)*{S}
\ar (4.75,-2.125);(4.75,-1.5)
\ar (4.5,-2.125);(4.125,-1.75)
% delta arrow
\ar (-0.625,-0.875);(-0.1,-0.15)
% A arrow
\ar (0.25,2.25);(0.25,3)
\ar@{--} (0.25,0.875);(0.25,1.875)
% B arrow
\ar@{--} (2.5,-1);(2.1,-1.8)
% C arrow
\ar (3.5,1.875);(5,1.875)
\ar (2,-2);(1.5,-3)
\end{xy}\]
\caption{}
\label{fig:6.4}
\end{figure}

But the total number of points $\left|\{q\}\right|$ along $S$ where the normal field vanishes (as it reverses direction) must be even.  If it were odd, $S$ would reverse orientation on the $R^{2k}$ normal bundle to its sheet in $T$.  But this normal bundle is a direct sum of two $R^k$ normal bundles: $S$ within $V$ and $S$ within $V^\prime$.  $S$ must either preserve or reverse the orientation of the normal $R^k$ in these two bundles together.  So the sum will certainly have its orientation preserved.  Thus,
\[0\equiv_2 x\cdot x = \left|\widetilde{V}\cap \widetilde{V}\right| \equiv_ 2 (a,Qb) + \left|\{q\}\right| \equiv_2 (a, Qb) + \text{even} = (a, Qb).\]
The first congruence is from the vanishing of $v_{k+1}\in H^{k+1}(X)$.
\end{proof}

\subsubsection*{Remark on hypothesis of Lemma \ref{lem:6.4}}
As initially introduced, Figure \ref{fig:6.1} portrayed two loops $l$ and $\alpha$ on the $2$-sphere $S^2$.  But let us now reinterpret Figure \ref{fig:6.2} as two $\left(S^1\times\mathbb{C}P^2\right)$s, $l$ and $\alpha$, within $S^2\times\mathbb{C}P^2$ by simply taking the Cartesian product of the entire diagram with $\mathbb{C}P^2$.  Thus $M$ is no longer four points but $\{4\text{ points}\}\times\mathbb{C}P^2$.  $X = S^2\times\mathbb{C}P^2$ has a nontrivial, low dimensional, Stiefel-Whitney class: $w_2(\tau(X))$ is the generator of $H^2\left(S^2\times\mathbb{C}P^2\right)$ corresponding to $H^2\left(\ast\times\mathbb{C}P^2\right)$ under the K\"{u}nneth decomposition. So Lemma \ref{lem:6.4} does not apply, and indeed the sign formula in Theorem \ref{thm:6.1} is not valid in this example.  Since $W_l$ produces isotopic ``before'' and ``after'' submanifolds $S^1\times\mathbb{C}P^2$, we expect the sign to be $+1$; however,
\[-1^{s(l)} i^{\chi(l\cap\alpha)} = -1^3 i^{12} = -1.\]\qed

Returning to the proof of Theorem \ref{thm:6.1}, consider the Mayer-Vietoris sequences for $B\cup C$, the initial multi-balloon; $A\cup C$, the final multi-balloon; and $A\cup B$, the Wilson balloon.  We use the abbreviations $a\cap c$, $b\cap c$, and $a\cap b$ for the dimensions of the intersected kernels $A\cap C$, $B\cap C$, and $A\cap B$, respectively.  The sequences read in part:

\be\label{eqn:seq1}\tag{$1$}
D\rightarrow R^{a\cap c}\rightarrow H_k(A\cap C)\rightarrow H_k(A)\oplus H_k(C)\rightarrow H_k(A\cup C)\rightarrow H_{k-1}(A\cap C)\rightarrow \cdots,
\ee

\be\label{eqn:seq2}\tag{$2$}
0\rightarrow R^{b\cap c}\rightarrow H_k(B\cap C)\rightarrow H_k(B)\oplus H_k(C)\rightarrow H_k(B\cup C)\rightarrow H_{k-1}(B\cap C)\rightarrow\cdots,
\ee

\be\label{eqn:seq3}\tag{$3$}
0\rightarrow R^{a\cap b}\rightarrow H_k(A\cap B)\rightarrow H_k(A)\oplus H_k(B)\rightarrow H_k(A\cup B)\rightarrow H_{k-1}(A\cap B)\rightarrow\cdots
\ee

A basic fact is that the alternating sum of dimensions in any exact sequence with field coefficients is zero.  We are only interested in collecting mod $2$ information so we consider just the non-alternating sum of dimensions.  let $\Delta = s(A\cup C) + s(B\cup C)$, the mod $2$ ``difference'' of semi-characteristics.  For notational convenience, we extend the definition of the semi-characteristic to odd dimensional manifolds with boundary such as $A$, $B$, and $C$, writing $s(A) = \sum_{i=1}^k b_i(A)\mod 2$, where $b_i$ is again the $Z_2$-Betti number.  From sequences \ref{eqn:seq1} and \ref{eqn:seq2} we see:
\[s(A\cup C) = a\cap c + \sum_{i=0}^k b_i(\delta) + \sum_{i=0}^k b_i(A) + \sum_{i=0}^k b_i(C)\text{, and}\]
\[s(B\cup C) = b\cap c + \sum_{i=0}^k b_i(\delta) + \sum_{i=0}^k b_i(B) + \sum_{i=0}^k b_i(C)\text{, so}\]
\be\label{eqn:6.3}
\Delta = a\cap c + b\cup c + s(A) + s(B)\text{, and by Lemma \ref{lem:6.3}} = \frac{b_k(\delta)}{2} + a\cap b + s(A) +s(B).
\ee
But from sequence \ref{eqn:seq3},
\be\label{eqn:6.4}
s(A\cup B) = a\cap b + s(A) + s(B) + \sum_{i=1}^k b_i(M).
\ee
But
\be\label{eqn:6.5}
\sum_{i=1}^k b_i(M) \equiv \frac{\chi(M)}{2} + \frac{b_k(M)}{2} \mod 2
\ee
so combining lines (\ref{eqn:6.3}), (\ref{eqn:6.4}), and (\ref{eqn:6.5}) we get
\be\label{eqn:6.6}
\Delta = s(A\cup B) + \frac{\chi(M)}{2}
\ee

Thus the change $\Delta$ in the semi-characteristic of a fluctuating hypersurface ($A\cup C$) upon $\operatorname{NOT}$ on the support ($A\cup B$) of the Wilson balloon operator is given by the right-hand side of  (\ref{eqn:6.6}), a local formula in agreement with our sign rule.

\end{proof}

\begin{note}
The computation of $x\cdot x$ above relies on a certain regularity of the cellular $k$-chain $V$ representing the class $x$.  We have implicitly used that $V$ is a pseudo-manifold, meaning that it is a union of ($k+1$)-cells glued together so that every $k$-face is on the boundary of precisely two ($k+1$)-cells.
\end{note}

\section{Excitations}
Given a $\GDS$ on some manifold, and given a ground state $\Psi$, the state $O\Psi$ may be an excited state for certain choices of the operator $O$. We now briefly discuss these excitations.

One interesting choice of the operator $O$
to consider is a product of Pauli $Z$ over all $(d-1)$-cells that intersect some arc on the dual lattice; we call these dual open Wilson loops.  This case is completely analogous to the case of the $\GTC$.  This operator commutes with all terms of the Hamiltonian {\it except} those terms near the two endpoints of the arc.  
The state $O\Psi$ is the same for any two arcs that are homologous.
The excitations created near the endpoints of the arc must occur in pairs: given such a ground state $\Psi$ and such an operator $O$, there is no state whose reduced density matrix agrees with $\Psi$ everywhere except near {\it one} of the endpoints of the arc, while agreeing with the reduced density matrix of $O\Psi$ near that endpoint.

We can also consider operators obtained from ``open balloons".  Assume $d=2k+2$, $k>0$.  Choose $L$ to be an ``open balloon" (a codimension $=1$ submanifold with boundary).  Let $\alpha$ be the $(d-1)$-dimensional cellular submanifold supporting a basis vector.
In this case, choose the operator to be $\pm \operatorname{NOT}$, where the sign is
given by
\[(-1)^{s(L)}\left(i^{\chi(\partial(L\cap\alpha))}\right).\] 
Note that in the DS model for $d=2$, the operators that drag semions $b$ and that drag semions $\overline b$ are related by complex conjugation; however, for the present theory, if we replace $i$ by $-i$ in the definition of the open balloon operator, then the resulting operator is unchanged because $\chi(\partial(L\cap\alpha))$ is even.

Balloon or open balloon operators will either commute or anti-commute with dual Wilson loop or dual open Wilson loop operators depending upon the parity of the intersection.
Further, dual open Wilson loop operators commute with each other.  The commutation relations between open balloons are more complicated, and we leave this as an open question.  We also leave the statistics of the excitations created by the open balloons as an open question.

\section{Comparison with $Z_2$-Dijkgraaf-Witten theories}
One aspect of the $\GDS$ theory (in $d=$ even dimensions) is that topological features of the $d$-manifold on which $X$ it lives may create ``missing sectors,'' elements $\alpha\in H_{d-1}\left(X;Z_2\right)$ so that for hypersurface $E\subset X$, $[E] = \alpha$, $\psi[E] = 0$ for any zgswf.  An example of this is $\C P^{2k}$.  Since $\chi\left(\C P^{2k}\right) = 2k+1=$ odd and $H_{4k-1}\left(\mathbb{C}P^{2k};Z_2\right) = 0$,
the Hilbert space $\GDS\left(\C P^{2k}\right) = 0$.

It is natural to wonder if there is a relation between $\GDS$ and the twisted $Z_2$-Dijkgraaf-Witten (DW) theory, and we thank Meng Cheng and Kevin Walker for raising this question and discussions on the resolution.  For $G=Z_2$, $\text{BG} = K\left(Z_2,1\right) = RP^\infty$, the infinite real projective space.
\[H^j\left(RP^\infty;R/Z\right)\cong\begin{cases} 0,& j\text{ even}\\Z_2,& j\text{ odd}\end{cases},\]
so there is a unique twisting to consider when $j$ is odd.  Generically the effect on the Hilbert space of a twisting is to introduce some relations into the corresponding untwisted space, so we should check if these relations could cause the ``missing sectors'' observed in $\GDS$.  Theorem \ref{thm:TDW} below shows this is \emph{not} the case.

In the untwisted DW theory for a finite group $G$, the Hilbert space $\DW\left(X^d\right) =$ is $\C$-linear formal combinations of isomorphism classes of $G$-principal bundles over $X$.  A twist is given by the choice of a cocycle $\alpha$, $[\alpha]\in H^{d+1}\left(\text{BG};R/Z\right)$.  (The class $[\alpha]$ determines the theory.)  Now, following \cite{FQ}, $\TDW\left(X^d\right)$ is generated by maps $f: X\rightarrow\text{BG}$ with a relation for every homotopy $F:X\times I\rightarrow\text{BG}$, $F\mid_{X\times 0} = f_0$ and $F\mid_{X\times 1}=f_1$ of the form $f_0=\omega f_1$, where $\omega = F^\ast(\alpha)[X\times I]$.  From this definition, it follows easily that if for all maps $g:X\times S^1\rightarrow\text{BG}$, $g^\ast(\alpha)\left[X\times S^1\right]=1$, then $\dim\TDW(X)=\dim\DW(X)$, i.e., there are no additional relations.  We use this criteria to prove:

\begin{theorem}\label{thm:TDW}
For $d$ even, $\GDS$ is distinct from $Z_2\TDW$ (and also $Z_2\DW$, the untwisted theory).
\end{theorem}
\begin{proof}
For simplicity, begin with $d=4k$, where $\GDS$ is distinguished from both twisted and untwisted Dijkgraaf-Witten theory by
\[\GDS\left(\C P^{2k}\right) = 0\]
and
\[Z_2\TDW\left(\C P^{2k}\right)\cong Z_2\DW\left(\C P^{2k}\right)\cong\C.\]
The first equality results from $\chi\left(\C P^{2k}\right)$, the Euler characteristic, being odd which implies that the empty hypersurface may be ``swept over'' $\C P^{2k}$ and back to itself with an odd number of Morse transitions, each of which counts for $(-1)$ in the wave function.

Since $\pi_1\left(\C P^{2k}\right)\cong 0$, there is only the trivial $Z_2$-principal bundle so
\[\DW\left(\C P^{2k}\right)\cong\C.\]
We will prove that any map $g:\C P^{2k}\times S^1\rightarrow RP^\infty$ homotopically factors through the second factor:
\[\vcenter{\vbox{\xymatrix{
g:\C P^{2k}\times S^1\ar[r]^-g \ar[d]^{\pi_2} & RP^\infty \\
S^1 \ar@{-->}[ur]^{g^\prime}
}}}\text{(commutes up to homotopy)}\]
This immediately implies (by a $Z_2$-Stokes theorem) that
\[g^\ast(\alpha)\left[\C P^2\times S^1\right] = {g^\prime}^\ast\left[S^1\right] = 1,\]
the second equality, since $S^1$ has no higher homology with any coefficient system.

Now we construct $g^\prime$.  Fix $\theta\in S^1$ and consider
\[\xymatrix{
&&S^\infty\ar[d]^{\times 2} \\
\C P^{2k} \ar@(ur,l)[rru]^l \ar[r]^-{\text{id}\times\theta} & \C P^{2k}\times S^1 \ar[r]^-g & RP^\infty
}\]
The lift $l$ exists since $\pi_1\left(\C P^{2k}\right) = 1$.  $l$ extends to a map over $\operatorname{cone}\left(\C P^{2k}\right)$ and this cone further projects back into $RP^\infty$.
\[\xymatrix{
\operatorname{cone}\left(\C P^{2k}\right) \ar[r]^-{\bar{l}} & S^\infty \ar[d]^{\times 2} \\
\C P^{2k} \ar@{^{(}->}[u]\ar[ur]^l & RP^\infty
}\]
Homotopically, this factors $g$ through $g_1:\Sigma\left(\C P^{2k}\right)\vee S^1\rightarrow RP^\infty$, $\Sigma$ denoting suspension.
\begin{figure}[hbpt]
\[g_1:\begin{xy}<5mm,0mm>:
(0,0)*\ellipse<15pt,7.5pt>{}
,(1.75,0)*{\C P^{2k}}
,(-2,0)*\ellipse<7.5pt,15pt>{}
,(2,0)*\ellipse<7.5pt,15pt>{}
,(-4,1.0625);(4,1.0625)**\dir{-}
,(-4,-1.0625);(4,-1.0625)**\dir{-}
,(-4.125,-1.0625);(-5.5,0)**\dir{-}
,(-4.125,1.0625);(-5.5,0)**\dir{-}
,(4.125,-1.0625);(5.5,0)**\dir{-}
,(4.125,1.0625);(5.5,0)**\dir{-}
,(-5.5,0);(5.5,0)**\crv{(-6,0.75)&(-6,3)&(6,3)&(6,0.75)}
,(-5.5,-0.75);(-5.5,-1.5)**\dir{-}
,(-5.5,-1.5);(5.5,-1.5)**\dir{-}
,(5.5,-1.5);(5.5,-0.75)**\dir{-}
,(0,-2.25)*{\Sigma\left(\C P^{2k}\right)}
,(0,2.875)*{\bullet}
,(0,2.25)*{\theta}
\end{xy}\longrightarrow RP^\infty\]
\caption{}
\label{fig:suspension}
\end{figure}
But $\Sigma\left(\C P^{2k}\right)$ is also simply connected so $g_1$ extends over a cone on $\Sigma\left(\C P^{2k}\right)$, to a map $g_2$,
\[g_2:\Sigma\left(\C P^{2k}\right)\vee S^1\cup\operatorname{cone}\left(\Sigma\left(\C P^{2k}\right)\right) \rightarrow RP^\infty.\]
But the source of $g_2$ is homotopy equivalent to $S^1$, and making this identification, $g_2$ becomes the desired $g^\prime$.

With this warmup we can now handle the case $d=4k+2$, $k\geq 1$.  This is done by considering the Hilbert spaces for $RP^2\times \C P^{2k}$.  The argument will be an iterative version of the previous.  Since $\chi\left(RP^2\times\C P^{2k}\right) = 2k+1=$ odd, it again will suffice to show that for $g:RP^2\times \C P^{2k}\times S^1\rightarrow RP^\infty$ homotopically factors through a map $g^\prime: RP^2\times S^1\rightarrow RP^\infty$ (since $H^{d+1}\left(RP^2\times S^1; R/S^1\right)\cong 0$).

To simplify notation, set $M=RP^2\times S^1$.  Initially we have a (product) bundle:
\[\xymatrix{
\C P^{2k}\ar[r] & \C P^{2k}\times M\ar[r]^-g\ar[d] & RP^\infty \\
& M\ar@{-->}[ur]^h
}\]
where the homotopy extension $h$ exists since $\pi_1\left(\C P^{2k}\right) = 1$.

Make $h$ transverse to $RP^{\infty-1}\subset RP^\infty$ with $Y = h^{-1}\left(RP^{\infty-1}\right)$, and write $M=Y \cup P$, $P=M\setminus Y$.  By the previous method, $g$ factors through $g_1$,
\[g_1:\left[\nu_Y\times \C P^{2k}\cup\left(\partial\left(\nu_Y\right)\times\operatorname{cone}\left(\C P^{2k}\right)\right)\right]\cup P\rightarrow RP^\infty,\]
where $\nu_Y$ denotes the normal bundle to $Y$ in $M$.  Notice the space in brackets is a bundle $E$ over $Y$ with a simply connected fiber: $\Sigma\left(\C P^{2k}\right)$:
\[\xymatrix{
\Sigma\left(\C P^{2k}\right) \ar[r] & E \ar[d] \ar[r]^-{g_1\mid} & RP^\infty \\
& Y \ar@{-->}[ur]^{h_1}
}\]
As before, $\pi_1\left(\Sigma\C P^{2k}\right) = 1$ implies a homotopy factoring through $h_1$.  Now make $h_1$ transverse to $RP^{\infty-1}\subset RP^\infty$ and set $h_1^{-1}\left(RP^{\infty-1}\right) = L\subset Y$.  Set $Y\setminus L = Q$, $Y = L\cup Q$.

Now $g_1$ factors through $g_2$,
\[g_2:\left[\left(\nu_L\times\Sigma\left(\C P^{2k}\right)\right)\cup \partial\left(\nu_L\right)\times\operatorname{cone} \left(\Sigma\left(\C P^{2k}\right)\right)\right] \cup Q\cup P\rightarrow RP^\infty.\]
Again the space in brackets is a bundle with simply connected fiber: $\Sigma\left(\Sigma\left(\C P^{2k}\right)\right)$:
\[\xymatrix{
\Sigma^2\left(\C P^{2k}\right) \ar[r] & F \ar[d] \ar[r]^-{g_2\mid} & RP^\infty \\ & L \ar@{-->}[ur]_{h_2}
}\]
Finally, make $h_2$ transverse to $RP^{\infty-1}\subset RP^\infty$ with $h_2^{-1}\left(RP^{\infty-1}\right) = J$ and $L\setminus J = K$, $L=J\cup K$.  Now $g_2$ factors through $g_3$,
\[g_3:\left[\left(\nu_J\times \Sigma^2\left(\C P^{2k}\right)\right)\cup\partial\left(\nu_J\right)\times\operatorname{cone}\left(\Sigma^2\left( \C P^{2k}\right)\right)\right]\cup K\cup Q\cup P\rightarrow RP^\infty.\]

Again the space in brackets is a bundle with simply connected fiber $\Sigma^3\left(\C P^3\right)$:
\[\xymatrix{
\Sigma^3\left(\C P^{2k}\right) \ar[r] & G\ar[d]\ar[r]^-{g_3\mid} & RP^\infty \\ & J
}\]
But now $J$ is merely a finite number of points indexing the components of $G$.  Each such component can now be lifted to $S^\infty$, coned in $S^\infty$, and then projected back to $RP^\infty$.  The result is to factor $g_3$ through $g_4$ where the source of $g_4$ is the source of $g_3$ with these cones attached to the components of $G$.  But with these cones the source of $g_4$ is homotopy equivalent to $RP^2\times S^1 = J\cup K\cup Q\cup P$.  Using this homotopy equivalence, $g_4$ may be identified with the map we have been seeking, $g^\prime$.
\end{proof}

\section{Discussion and Relation to Other TQFTs}
We have constructed a Hamiltonian which generalizes the double semion TQFT to higher dimensions and we have computed the dimension of the space of zero energy states for several different choices of ambient space.  For odd dimensions $d$, we have shown (Remark \ref{qcremark}) that the ground state subspace of the theory is related that of
the toric code by a local unitary transformation.

For even dimensions $d$, we conjecture that: for all non-empty choices of ambient space,
there is no local unitary transformation that maps the ground state subspace of the GDS theory into the ground state subspace of the GTC or twisted $Z_2$ Dijkgraaf-Witten model.
More strongly, we conjecture that: given a ball $B$ contained in the ambient space,
 no ground
state wavefunction of the GDS
theory can be mapped by a local unitary transformation in such a way that the reduced density matrix on $B$
agrees with the reduced density matrix on $B$
of either a ground state wavefunction of the GTC or a ground state wavefunction of the twisted $Z_2$ Dijkgraaf-Witten model.  To make this conjecture more precise, one may consider a sequence of cellulations of the ambient space with bounded local
geometry and increasing number of cells so that the diameter of the cells tends to zero; then one constructs a sequence of quantum circuits which map some ground state wavefunction of the GDS theory so that the reduced density matrix on $B$ has the desired property.  These circuits are chosen to minimize the depth of the circuit (taking the gates in
the circuit to have bounded range); then the conjecture is that this depth diverges.

This conjecture is one way of formally stating that the Hamiltonians correspond to different phases of matter.
The main
evidence for this belief in even dimensions is that we have shown that for certain choices of ambient space the ground state degeneracy of the GDS model does not agree with that of either the toric
code or the twisted $Z_2$ Dijkgraaf-Witten model; this immediately implies that there is no local unitary transformation that maps the ground subspace of one theory onto
that of the other theory in those cases as the dimensions are different.  However, more work is needed to establish the stronger conjecture above.  This is perhaps reminiscent of the case in two dimensions where it was shown first that no ground state of the toric code on a torus can be mapped by a local quantum circuit to a product state\cite{bhv}, then it was shown that the ground state of the toric code on a sphere could not be mapped
to a product state\cite{leshouches}, and finally more general invariants were found showing that the topological S-matrix of the theory is invariant under local unitary transformations\cite{haah}.  Hence, to show such a conjecture here may require a more detailed analysis of excitations of the theory.  This
may require a description in terms of higher categories\cite{kongwenzheng}.

\appendix
\section{Generic Cellulations}
Consider any set $S$ of points in general position in $d$-dimensions.  Here ``general position" means that no more than $d+1$ points lie on any hypersphere and that
no more than $n+1$ points lie on any $n$-dimensional hyperplane for $n<d$.
Consider the Voronoi cells defined by this set of points.
Consider any subset $T$ of the $d$-dimensional Voronoi cells.  Let $B$ be the boundary of this subset.  We now show that $B$ is a $d-1$-dimensional PL manifold.

Let $x$ be any point in $B$.  Then, $x$ is at some distance $D$ from some set of points $p_1,p_2,...,p_k$ in which are centers of cells in $T$ and points $q_1,...,q_l$ which are centers of cells not in $S$.  Thus, there are $k+l$ points in $S$ at distance $D$ from $x$.  Because the points in $S$ are in general
position, $k+l\leq d+1$.
Further, $k>0$ and $l>0$ (otherwise, $x$ is not on the boundary), so $k,l \leq d$.

For notational simplicity, translate the point $x$ to the origin so $|p_i|=|q_j|=D$.
The points in $B$ near the origin are the points that are the same distance from centers of cells in $T$ as from centers of cells not in $T$.
For any point $y$, the distance from $y$ to a point $p_i$ is equal to $\sqrt{d^2-2 p_i \cdot y + y^2}$.
So, a point $y$ near the origin (i.e., sufficiently close to the origin that $p_1,...,p_k,q_1,...,q_l$ are closer to $y$ than any other point in $S$)
is in $B$ if and only if
\be
\label{equaldist}
{\rm max}_i p_i \cdot y = {\rm max}_j q_j \cdot y.
\ee

Since the points are in general position, the vectors $p_i,q_j$ span at least $k+l-1$ dimensions.  Suppose that the vectors $p_i,q_j$ span $k+l-1$ dimensions.  Then, the
set of points satisfying the Eq.~(\ref{equaldist}) is $R^{d+1-k-l}$ times the set of points in a $k+l-1$ dimensional space (i.e., the space which is the span of $p_i,q_j$) which satisfy Eq.~(\ref{equaldist}).  So, in this case we can reduce the problem of showing that $B$ is a PL manifold to the problem of showing that $B$ is a PL manifold in this case that $k+l=d+1$.

Suppose instead that the vectors $v$ span $k+l$ dimensions.
Then, there is some nonzero vector $v$ such that $p_i \cdot v = q_j \cdot v={\rm constant}$ for all $i,j$.  Then, translating the origin by this vector $v$, we arrive at a new problem in which the vectors $p_i,q_j$ span
only $k+l-1$ dimensions.  This reduces to the above case where $p_i,q_j$ span $k+l-1$ dimensions (this vector $v$ is a vector in the space $R^{d+1-k-l}$ above).
So, in general we can assume that $k+l=d+1$ and the points span $d$ dimensions.

Consider the space of vectors $v$ such that
\be
\label{peq}
v \cdot p_i = v \cdot p_j
\ee
for all $i,j$
and
\be
\label{qeq}
v \cdot q_i = v \cdot q_j
\ee
for all $i,j$.  These equations give $(k-1)+(l-1)=d-1$ constraints.  Hence, there is at least a one dimensional space of vectors satisfying Eqs.~(\ref{peq},\ref{qeq}).
Let $v$ be a nonzero vector satisfying Eqs.~(\ref{peq},\ref{qeq}) with $v\cdot p_i=z_p$ for some constant $z_p$ and $v \cdot q_j = z_q$ for some constant $z_q$.
By the assumption that the points are in general position, $z_p \neq z_q$, as otherwise there would be $d+1$ points on the hyperplane of fixed inner product with $v$.

There is a $d-1$ dimensional space of points orthogonal to $v$.  Let $m$ be any vector in this space.
Then,
\be
y=m+cv
\ee
satisfies Eq.~(\ref{equaldist}) if and only if
\be
c= \frac{{\rm max}_i  p_i \cdot m- {\rm max}_j  q_j \cdot m}{z_q-z_p}.
\ee
So, setting $y=m+\frac{{\rm max}_i  p_i \cdot m- {\rm max}_j  q_j \cdot m}{z_q-z_p}v$
gives a PL homeomorphism from a neighborhood of the origin in $R^{d-1}$ to a neighborhood of the origin in $B$.

Note that without the assumption of general position, $B$ need not be a manifold.  As an example consider four points in the plane, at coordinates $(\pm 1, \pm 1)$.  Consider the set of Voronoi cells containing the cells corresponding to the points at $(+1,-1)$ and $(-1,+1)$.  Then, the boundary of this set is the horizontal and vertical axes, which has a singularity at the origin.

\begin{thebibliography}{99}
\bibitem{DW} R. Dijkgraaf and E. Witten, ``Topological gauge theories and group cohomology",
Comm. Math. Phys.{\bf 129}, 393-429 (1990).

\bibitem{K}  A. Kitaev, ``Fault-tolerant quantum computation by anyons", Ann. Phys. {\bf 303}, 2 (2003).

\bibitem{systolic}  M. Freedman, D. Meyer, and F. Luo. ``$Z_2$-systolic
freedom and quantum codes", in {\it Math. of Quantum Computation}, 
pages 287–320, eds. R. K. Brylinski and
G. Chen, CRC Press (Boca Raton) 2002.

\bibitem{dennis} E. Dennis, A. Kitaev, A. Landahl, and J. Preskill,
``Topological quantum memory", J. Math. Phys. 43, 4452-4505 (2002).

\bibitem{LW} M. A. Levin and X.-G. Wen, ``String-net condensation: A physical mechanism for topological phases",
Phys.Rev. B {\bf 71}, 045110 (2005).

\bibitem{Q} F. Quinn, ``Lectures on Axiomatic Topological Quantum Field Theory", in {\it Geometry and Quantum Field Theory}, page 323, eds. D. S. Freed and K. K. Uhlenbeck, IAS/Park City Mathematics, Vol 1,  American Mathematical Society, IAS (Princeton, NJ) 1991.

\bibitem{WW} K. Walker and Z. Wang, ``(3+1)-TQFTs and Topological Insulators", Front. Phys. {\bf 7}, 150 (2012), arXiv:1104.2632.

\bibitem{stab1} S. Bravyi, M. B. Hastings, and S. Michalakis,
``Topological quantum order: stability under local perturbations",
J. Math. Phys. {\bf 51} 093512 (2010).

\bibitem{stab2} S. Bravyi and M. B. Hastings,
``A short proof of stability of topological order under local
 perturbations",Commun. Math. Phys. {\bf 307}, 609 (2011).

\bibitem{FQ} D.~Freed and F.~Quinn. \emph{Chern-Simons Theory with Finite Gauge Group}, Comm. Math. Phys. 156 No. 3, 435---472 (1993).

\bibitem{bhv} S. Bravyi, M. B. Hastings, F. Verstraete, ``Lieb-Robinson bounds and the generation of correlations and topological quantum order",
Phys. Rev. Lett. 97, 050401 (2006).

\bibitem{leshouches} M. B. Hastings, ``Locality in Quantum Systems", arXiv:1008.5137, Les Houches Lecture Notes, in {\it Quantum Theory from Small to Large Scales, Volume 95, 2010}, eds. J. Frohlich, M. Salmhofer, V. Mastropietro, W. De Roeck, and L. F. Cugliandolo, Oxford University Press (Oxford), 2012.

\bibitem{haah} J. Haah, ``An invariant of topologically ordered states under local unitary transformations", 
arXiv:1407.2926.

\bibitem{kongwenzheng} L. Kong, X.-G. Wen, and H. Zheng, ``Boundary-bulk relation for topological orders as the functor mapping higher categories to their centers", arXiv:1502.01690.
\end{thebibliography}
\end{document}